\documentclass[11pt,a4paper]{article}
\usepackage{xr-hyper}

\usepackage{amsmath}
\usepackage{amsthm}
\usepackage{amssymb}
\usepackage{amsfonts}
\usepackage{wasysym} 
\usepackage{cancel}
\usepackage{dsfont} 
\usepackage{subfiles}
\usepackage{mathrsfs}
\usepackage{mathtools}
\usepackage{graphics}
\usepackage[utf8]{inputenc}

\usepackage{soul} 

\usepackage[a4paper,verbose,bottom=3.0cm,top=2.5cm,left=3.2cm,right=3.2cm]{geometry}

\usepackage{scalerel,stackengine}
\stackMath
\newcommand\reallywidecheck[1]{%
\savestack{\tmpbox}{\stretchto{%
\scaleto{%
\scalerel*[\widthof{\ensuremath{#1}}]{\kern-.6pt\bigwedge\kern-.6pt}%
{\rule[-\textheight/2]{1ex}{\textheight}}
}{\textheight}%
}{0.5ex}}%
\stackon[1pt]{#1}{\scalebox{-1}{\tmpbox}}%
}

\usepackage{alltt}

\newcommand{\fl}[1][3em]{\hspace{#1}&\hspace{-#1}} 

\usepackage[htt]{hyphenat}
\usepackage{xparse}
\usepackage{microtype}
\usepackage{color}
\usepackage[pdftex,pdfborderstyle={/S/U/W 0},colorlinks=true,linkcolor=blue,pdfnewwindow=true]{hyperref} 
\usepackage{cleveref}

\hypersetup{linkcolor=blue!30!black!80,citecolor=blue!30!black!80,urlcolor=blue!30!black!80}

\usepackage{tikz}

\usepackage{etoolbox}

\usepackage{enumerate}

\newcommand{\nodate}{\date{\vspace{-1.5em}}}

\newcommand{\I}{\mathrm i} 
\newcommand{\Ii}{\mathrm i} 
\newcommand{\Ee}{\mathrm e} 
\newcommand{\Alg}{\mathfrak A}

\newcommand{\term}[1]{{\it #1}\/}

\newcommand{\Reg}{\mathcal O}

\newcommand{\DoubleCone}{{\mathscr C}}

\newcommand{\PoincareGroup}{\mathcal P}
\newcommand{\LorentzGroup}{\mathcal L}

\newcommand{\RealNum}{\mathbb R}
\newcommand{\ComplexNum}{\mathbb C}

\newcommand{\NaturalNum}{\mathbb N}
\newcommand{\SchwartzSpace}{\mathscr S}
\newcommand{\DInt}[2][\,]{{\mathrm d}^{\hspace{-0.20ex}#1}\hspace{-0.25ex}#2}

\newcommand{\norm}[1]{\left \|#1 \right\|}
\newcommand{\normm}[1]{\|#1\|}
\newcommand{\abs}[1]{\left |#1 \right|}
\newcommand{\abss}[1]{|#1|}

\newcommand{\HilbertSpace}{\mathscr H}
\newcommand{\BoundedOps}{\mathrm B}

\newcommand{\LSpace}{L}

\newcommand{\Id}{\mathds 1}

\newcommand{\ContinuousFuncs}{\mathnormal C}

\newcommand{\IndicatorFunction}{\mathds 1}

\newcommand{\CharFct}{\IndicatorFunction}

\renewcommand{\vec}[1]{{\bf #1}} 

\newcommand{\closure}[1]{\overline{#1}}

\DeclareMathOperator*{\support}{supp}

\renewcommand{\Re}{\operatorname{Re}}

\DeclareMathOperator{\Dom}{D} 

\newcommand\numberthis{\addtocounter{equation}{1}\tag{\theequation}}

\newcommand{\standardThms}{\newcounter{thm}

\newtheorem{Prop}[thm]{Proposition}

\newtheorem{Lem}[thm]{Lemma}
\newtheorem{Cor}[thm]{Corollary}
\newtheorem*{Cor*}{Corollary}
\newtheorem{Thm}[thm]{Theorem}
\newtheorem*{Thm*}{Theorem}

\newtheorem{Def}[thm]{Definition}

\theoremstyle{remark}
\newtheorem{Rem}[thm]{Remark}

\newcommand{\solution}{\proof[Solution.]}
}

\def\WO{\mathbf W}

\DeclareMathOperator{\Span}{span}

\def\Pp{{\bf P}}
\standardThms

\def\Ww{\mathcal W} 
\def\WO{\mathbb W} 
\def\WwR{{\Ww_{\mathrm r}}} 
\def\Wwr{\WwR} 
\def\T{\tau}
\def\Pp{\boldsymbol P}
\def\s{s}

\newcommand{\HCol}[1]{{#1}}
\newcommand{\Aa}{\HCol A}
\newcommand{\Bb}{\HCol B}

\newcommand{\BB}{\HCol{\mathcal B}}

\def\itref#1{(\ref{#1})}
\def\FockSpace{\Gamma}
\def\VelocitySupport{{\mathcal V}}
\def\FS{\FockSpace}
\def\VS{\VelocitySupport}
\def\vVS{\vec V}

\def\MAlg{{\mathfrak M}}
\def\C{{\mathrm c}}

\usepackage{amsopn}
\def\fin{{\mathrm f}}
\def\ini{{\mathrm i}}

\def\dd{d} 
\def\FSu{\FS^{\mathrm u}}
\def\BB{B}
\def\SO{\mathrm{SO}}
\def\OO{\mathrm{O}}

\title{
  \(N\)-Particle Scattering in Relativistic
  Wedge-Local 
\\
  Quantum Field Theory
  }
\def\mail{duell@ma.tum.de}
\def\affil{Zentrum Mathematik, Technische Universität München}
\author{Maximilian Duell\footnote{E-mail: \mail}\\[0.5em]\affil}
\nodate

\begin{document}
\maketitle

 \begin{abstract}
  Multi-particle scattering states are constructed for massive Wigner particles
  in the general operator-algebraic setting of wedge-local quantum field theory.
  The apparent geometrical restriction of the conventional wedge-local
  Haag-Ruelle argument to two-particle scattering states is overcome with a
  swapping symmetry argument based on wedge duality.
 \end{abstract}

\section{Introduction}
Wedge locality has become an increasingly prominent concept in mathematical
physics ever since wedge duality was established in the Wightman framework by
Bisognano and Wichmann \cite{BW75}. In particular, while interacting local
quantum field theories (QFT) in four dimensions are still missing, non-trivial
wedge-local QFT have emerged in recent years \cite{GL07,BLS10}. This provides
strong motivation to develop \(N\)-particle scattering theory in the wedge-local
setting, which is the goal of the present paper.

 The classical Wigner particle concept can still be consistently formulated in
 wedge-local theories as it does not depend on 
 any notion of localization in configuration space.
 Accordingly we may define massive single particle states
 \(\Psi_1 \in \HilbertSpace\) as eigenvectors corresponding to positive eigenvalues of the
 relativistic mass operator~\(M := \sqrt{H^2 - \Pp^2}\). 
 Two-particle scattering states were then constructed in \cite{GL07,BS08} along
 the lines of Haag-Ruelle, using that two particles can be separated by two
 wedge regions \cite{BBS01}.
 Scattering states with a larger number of particles however appeared
 inaccessible or even unnatural in the wedge-local setting as a result of a
 simple geometric consideration: it is impossible to write down three or more
 wedge-local operators whose localization regions are space-like separated.

 In this paper we give a construction of scattering states for an arbitrary
 number of massive Wigner particles in the general wedge-local setting.
 Underlying our arguments is a simple {\em swapping} symmetry, which follows
 from wedge duality and augments cyclicity
 of the vacuum~\(\Omega\) for wedge
 algebras. It states that for 
 a given wedge-local bounded operator~\(A\in \Alg(\Ww)
 \subset \BoundedOps(\HilbertSpace)\) localized in a wedge~\(\Ww \subset
 \RealNum^d\)
 there exists\footnote{up to technical points discussed in \Cref{sec:swap}}
   \(A^\perp\in \Alg(\Ww^\perp)\) such that
 \[ 
   A\Omega = A^\perp \Omega, 
   \label{eq:introSwap} \numberthis
 \]
 where \(A^\perp\) is localized in a translate~\(\Ww^\perp := \Ww' + x\), \(x
 \in \RealNum^\dd\), of the causal complement~\( \Ww' \) in
 Minkowski space of dimension~\(d=s+1\).
 The symmetry~\eqref{eq:introSwap} itself has been known for some time in the
 context of integrable models\footnote{Swapping relations are
   mentioned e.g.\ in \cite{BS08} above Thm. 3.2 for bounded operators, in 
   \cite{Le03} below (3.13) for wedge-local fields, and indirectly
   in even earlier works of Schroer. The general connection
   to wedge-duality has been investigated in depth by Borchers \cite{Bor95},
   Rem.~1.1 and subsequent comments.
 },
 but its utility for the construction of scattering states seems to have so far
 escaped the attention of the experts. In fact, its application 
 in scattering theory appears very natural from the perspective of the causal
 geometry of wedge regions.


Let us now explain the role of the swapping relation~\eqref{eq:introSwap}
for scattering theory by sketching the convergence argument as an example. Let
us recall the standard definitions of Haag-Ruelle theory by
selecting~\(A_k\in\Alg(\Ww)\) (\(1 \leq k \leq n\)) with non-vanishing
projection~\(\Psi_1 = E_{\{M = m\}} A_k \Omega \) onto one-particle space of
mass \(m > 0\) and smear their space-time 
translates~\(\alpha_x(A_k) := U(x) A_k U(x)^*\) first with an auxiliary Schwartz
function~\(\chi \in \SchwartzSpace(\RealNum^{\dd})\) and afterwards with a
positive-energy Klein-Gordon solutions \(f_k\) (also for mass \(m\)) to obtain
{\em creation-operator approximants}
\[ 
  \Bb_k := \Aa_k(\chi) := \int \DInt[\dd] x \; \chi(x) \alpha_x(\Aa_k),
  \label{eq:defCreat}\numberthis
\]
\[
  \BB_{k\T}(f_k) := \int \DInt[\s] x\; f_k(\T, \vec x) \alpha_{(\T, \vec x)}(\Bb_k), \quad
  (\T \in \RealNum).
  \label{eq:defHR}\numberthis
\]
The smearing operation~\eqref{eq:defCreat} suitably restricts the
energy-momentum transfer, while \eqref{eq:defHR} may be understood as a
comparison dynamics in the sense of scattering theory. 
More precisely due to mass gaps we may arrange 
\(\Bb_k\Omega \in E_{\{M = m\}}  \HilbertSpace\) 
for suitable~\(\chi\) (supported in a sufficiently small neighbourhood of the mass shell)
and then \(\BB_{k\T}(f)\Omega = \tilde f_k(\Pp)\Bb_k\Omega\) is
a one-particle state created from the vacuum, which is independent of the
parameter~\(\T \in \RealNum\).  Scattering states are now to be constructed as
limits 
  \[
    \Psi^+ := \lim_{\T \to \infty} \Psi_\T, \quad
    \Psi_\T := \BB_{1\T}(f_1) \BB_{2\T}(f_2) \ldots \BB_{n\T}(f_n)\Omega,
      \label{eq:scattApprox}\numberthis
  \]
whose existence can be reduced to the one-particle convergence if the
norm of pairwise commutators is sufficiently decaying in \(\T\).
However, even if the Klein-Gordon solutions~\(f_k\) describe wave packets which
separate for large enough \(\T\to \pm \infty\), we should not expect
\(\BB_{k\T}(f_k)\) to commute in a general wedge-local model.

Here the swapping relation~\eqref{eq:introSwap} enters and
yields a second family of creation operators defined analogously in terms of
\(\Aa_k^\perp\)  which satisfy
\[\BB_{k\T}^\perp(f_k)\Omega =
  \BB_{k\T}(f_k)\Omega.
  \label{eq:swapCreat}\numberthis
\]
Across the two operator families we now obtain for suitably propagating wave
packets~\(f_k\) an asymptotic decay
\[
  \norm{\left[ \BB_{j\T}(f_j), \BB_{k\T}^\perp(f_k) \right]}
     \leq C_N  (1+\T)^{-N} \quad
  \text{for \(1\leq j < k \leq n\), \(\T > 0\)}.
  \label{eq:commutIntro}\numberthis
\]
To establish convergence of \eqref{eq:scattApprox} we estimate via Cook's method (\(0 < \T_1 < \T_2\))
\begin{align*}
  \norm{\Psi_{\T_2} - \Psi_{\T_1}} &= \norm{ \int_{\T_1}^{\T_2} \DInt \T \, \partial_\T \Psi_\T }
  \leq   \int_{\T_1}^{\T_2} \DInt \T \, \norm{ \partial_\T \Psi_\T },
  \label{eq:introCook}\numberthis
\end{align*}
where the integrand on the right hand side is expanded using the product rule.
To estimate the resulting terms  we make use of \eqref{eq:swapCreat}
 to write
\begin{align*}
  \fl
  \BB_{1\T}(f_1)  \ldots (\partial_\T \BB_{k\T}(f_k)) 
  \ldots \BB_{n\T}(f_n)\Omega
  \\
  &=
  \BB_{1\T}(f_1)  \ldots (\partial_\T \BB_{k\T}(f_k))
  \ldots 
  \BB_{n-1\,\T}(f_{n-1})
  \BB_{n\T}^\perp(f_n)\Omega
  \\&=
  \BB_{n\T}^\perp(f_n)
  \BB_{1\T}(f_1)  \ldots (\partial_\T \BB_{k\T}(f_k)) 
  \ldots 
  \BB_{n-1\,\T}(f_{n-1}) \Omega
  \\&\qquad \qquad
  + \text{(commutators)},
\end{align*}
where commutator terms vanish rapidly  as \(\T \to \infty\) by
\eqref{eq:commutIntro}, \(\norm{\BB_{j\T}(f_j)} \leq C (1+\abs{\T}^{s/2})\)
and \(\normm{\BB_{j\T}^\perp(f_j)} \leq C (1+\abs{\T}^{s/2})\). Iterating a total of \(n-k\) times,
the derivative term will act directly on the vacuum so that we can make use of 
\(\partial_\T \BB_{k\T}(f_k) \Omega = 0\) as in standard Haag-Ruelle theory.  
Altogether 
\eqref{eq:commutIntro} and polynomial norm growth of \(\BB_{j\T}(f_j)\),
\(\BB_{j\T}^\perp(f_j)\)
yield for \(\T > 0 \) the rapid decay 
 \begin{align*}
   \norm{
  \BB_{1\T}(f_1)  \ldots (\partial_\T \BB_{k\T}(f_k)) 
   \ldots \BB_{n\T}(f_n)\Omega} \leq C_N' (1+ \T)^{-N}.
 \end{align*}
 Summing up these terms, we obtain convergence of outgoing scattering
 states~\(\Psi^+\) from Cook's method~\eqref{eq:introCook}. A similar swapping
 argument yields the Fock structure of these scattering states for any number of
 particles~\(n \geq 0\). 
 For \(n\leq 2\) swapping is strictly speaking not necessary, as scattering states can be directly
 constructed via 
 \(
   \lim_{\T \to \infty} \BB_{\T}(f) \BB_{\T}^\perp(f^\perp)\Omega  
 \)
 as in \cite{BBS01,GL07}.
 Lastly it is important to point out that beyond swapping, it is also necessary
 that all operators~\(A_k\) entering in \eqref{eq:scattApprox} are 
 localizable in a common wedge \(\Ww\). Further, the propagation velocities of
 \(f_k\) must be suitably restricted to match the wedge geometry and be in
 correspondence with the fixed ordering of creation-operator approximants
 in~\eqref{eq:scattApprox}, as will be made precise in Sections~\ref{sec:scatt}
 and~\ref{sec:genW}.

Our construction applies in particular to the model
of Grosse and Lechner~\cite{GL07}. This model originated 
from a proposed quantum field theory on a non-commutative space-time, which may
be motivated from gravitational considerations~\cite{DFR95}. Only later a
reinterpretation as wedge-local quantum field theory on ordinary Minkowski
space-time was discovered and it was shown that this model exhibits non-trivial
2-particle scattering~\cite{GL07}.
The curious message of \cite{GL07} was that the model itself is
Poincaré-covariant, while Lorentz symmetry is broken at the level of scattering
states.  To clarify this effect, which is impossible in local quantum field
theories, we give a general analysis of Poincaré covariance of the scattering
states in \Cref{sec:wave}. We intend to apply these results to extend the
pioneering analysis of Grosse and Lechner to the multi-particle scattering data
in a subsequent publication.

This paper is structured as follows.
In \Cref{sec:framework} we introduce the wedge-local variant of the Haag-Kastler
framework providing the standing assumptions of our construction. The
wedge-local Haag-Ruelle theorem is established in \Cref{sec:scatt} under
certain geometrical restrictions allowing for a streamlined proof.
These restrictions are lifted in \Cref{sec:genW}, where we also obtain residual
Lorentz covariance properties and pave the ground for a general discussion of
wave operators and S-matrices in \Cref{sec:wave}.

\section*{Acknowledgements}
I am deeply indebted to Wojciech Dybalski for many valuable suggestions
and his continuous support. Further I would like thank Detlev Buchholz for
comments and communicating \Cref{lem:bu}, and Daniela Cadamuro for helpful
discussions.  I also gratefully acknowledge funding by the DFG within grant
DY107/2-1.

\section{Wedge-Local Quantum Field Theories} 
\label{sec:framework}
Our results are valid for Quantum Field Theory models defined on general
Minkowski space-time~\(\RealNum^{d}\), whose metric we take in the mainly-minus
convention and whose spatial dimension we denote by~\(s:= d-1\).
The family of wedge regions is defined as the
orbit~\(\PoincareGroup \Wwr:=\{\lambda \Wwr = \Lambda \Wwr + x, \lambda  =
    (x,\Lambda) \in \PoincareGroup \} \) of the conventional \term{right
    wedge}~\(\Wwr := \{ (t, \vec x) \in \RealNum^{d}: \abs{t} < x^1 \} \) 
    under the action of the Poincaré group~\(\PoincareGroup\) \cite{BW75}.

A wedge-local quantum field theory model 
in operator-algebraic formulation 
is specified by mathematical objects~\((\Alg, \alpha, \HilbertSpace,
\Omega)\), where \(\HilbertSpace\) is the Hilbert space of pure states
containing the \term{vacuum} as a distinguished unit vector \(\Omega \in
\HilbertSpace\).
The {\em wedge-local net} \(\Alg\) is a mapping from the family wedge regions \(\PoincareGroup\Wwr
\ni \Ww\) to von Neumann algebras~\(\Alg(\Ww) \subset
\BoundedOps(\HilbertSpace)\), which serves to describe Einstein causality at the quantum
mechanical level.  Poincaré symmetry acts on the wedge-local net~\(\Alg\) by a given group of
isomorphisms\footnote{The formulation of our main results requires only
  space-time translations. With some abuse of notation we denote translation
  automorphisms by the same letter \(\alpha\), or \(\alpha_x\), where \(x \in
  \RealNum^d\) is identified with \(\lambda_x = (x, \Id) \in
\PoincareGroup_+^\uparrow\). 
  In particular the basic version of the framework given by
\eqref{eq:HK1}--\eqref{eq:HK6} suffices for multi-particle scattering provided a
suitable swapping assumption holds, and we will state explicitly when the
strengthened variants~\eqref{eq:HK2s} or \eqref{eq:HK3s} are required.
}
\(\alpha_\lambda\) and we denote by 
\(\lambda = (x,\Lambda) \in
  \PoincareGroup_+^\uparrow = \RealNum^d \rtimes \LorentzGroup_+^\uparrow\)
  the elements of the proper orthochronous Poincaré group.

Guided by physical intuition we ask that these objects satisfy wedge-local
variants of the Haag-Kastler postulates, which are concerned with the algebraic
and representation-theoretic properties of \(\Alg\).
Firstly, for any choice of wedge regions~\(\Ww, \Ww_1, \Ww_2 \) we have
\newcommand{\lt}[1]{\text{\bf #1}}%
\begin{align*}
  \lt{Isotony} \quad 
    & \Alg(\Ww_1) \subset \Alg(\Ww_2) \text{ for } \Ww_1 \subset \Ww_2,
  \tag{HK1} \label{eq:HK1}\\
  \lt{Locality} \quad 
  & \Alg(\Ww_1) \subset \Alg(\Ww_2)' 
\text{ for } \Ww_1 \subset \Ww_2',
\tag{HK2} \label{eq:HK2}\\
\lt{Wedge-Duality}  \quad 
  & \Alg(\Ww') = \Alg(\Ww)',
\tag{HK2$^\sharp$} \label{eq:HK2s}\\
  \lt{Translation-Covariance} \quad   
  & \alpha_x(\Alg(\Ww)) = \Alg(\Ww +x), \quad x \in \RealNum^{d},\tag{HK3} \label{eq:HK3}\\
  \lt{Poincaré-Covariance} \quad   
  & \alpha_\lambda(\Alg(\Ww)) = \Alg(\lambda\Ww), \quad \lambda 
  \in \PoincareGroup_+^\uparrow. \tag{HK3$^\sharp$} \label{eq:HK3s}
\end{align*}
Here the Minkowski causal complement  \(\Ww' = - \Lambda \WwR + x\) of 
\(\Ww\) is also a wedge region and \(\Alg(\Ww)'\) denotes the commutant of
\(\Alg(\Ww)\) relative to \(\BoundedOps(\HilbertSpace)\).

On the representation-theoretic side we further assume that translations are
unitarily implemented on the vacuum Hilbert space~\(\HilbertSpace\) by a
strongly continuous \(\s\!+\!1\)-parameter group, \(\alpha_x(A) = U(x) A
  U(x)^*\). The representing unitaries are generated 
  by the \term{energy-momentum operators}
  via~\(U(x) = U(t, \vec x) = \Ee^{\Ii t H - \Ii \vec x \cdot \Pp}\),
  whose joint spectral resolution in terms of projection-operator-valued measures
  will be denoted by \(\Delta \longmapsto E(\Delta)\).
Focusing  also in particular on the analysis of scattering theory it will be convenient
to further impose the following standard assumptions concerned with the vacuum
representation and its one-particle spectrum,
\begin{align*} 
  \lt{ Uniqueness of \(\boldsymbol \Omega\)}
      \quad  & E(\{0\}) \HilbertSpace = \ComplexNum \Omega,
        \tag{HK4} \label{eq:HK4}\\
    \lt{Cyclicity of \(\boldsymbol \Omega\)} 
      \quad & \closure{\Alg(\Ww)\Omega} = \HilbertSpace, 
        \tag{HK5} \label{eq:HK5} \\
    \lt{Mass Gap}
      \quad & H_m \subset \support E \subset \{0\} \cup H_m \cup \bar H_M \subset \bar V^+,
        \tag{HK6} \label{eq:HK6}
\end{align*}
for some \(M > m > 0\), where \(H_m := \{ (\omega_m(\vec p) , \vec p),\, \vec p
\in \RealNum^\s \}\),
\( \omega_m(\vec p):= \sqrt{\vec p^2 + m^2}\), is the (positive) hyperboloid of mass~\(m >
0\) and \(\bar H_M := \{ (\omega, \vec p),\, \vec p
\in \RealNum^\s, \omega \geq \omega_M(\vec p) \}\) denotes the convex hull of
\(H_M\). Note that \eqref{eq:HK6} implies in particular that the one-particle subspace
\(\HilbertSpace_1\) and the corresponding orthogonal projection~\(E_m
:= E(H_m)\) are non-trivial.
We may extended any given wedge-local net also to
regions obtained as sum of a given wedge and any open bounded region~\(\Reg \subset
\RealNum^{\s+1}\) 
by setting \(\Alg(\Reg+\Ww) := (\cup_{x \in \Reg} \Alg(\Ww + x))''\).

For later convenience we will also introduce some refined terminology for wedge
regions concerning their geometry in the case of more than two dimensions. 
Recalling that any wedge region can be written as \(\Ww = \Lambda \WwR + x\), 
we may define the corresponding centered wedge as \(\Ww_\C := \Lambda \WwR\). 
\(\Ww_\C\) is uniquely characterized by the coordinate origin being contained in
its edge, and we will call such wedges \term{centered}.
This concept may be motivated heuristically by noting that scattering situations
are concerned with phenomena at very large distances, making finite translation
by \(x \in \RealNum^d\) in a sense negligible. 
Centered wedges~\(\Ww\) are convex cones in the sense that \(\Ww + \Ww \subset
\Ww\). This assures that the causal ordering given via the \term{precursor}
relation \cite{BBS01}
\[
  \Reg_1 \prec_\Ww \Reg_2 \Longleftrightarrow \Reg_2 - \Reg_1 \subset \Ww_\C
  \label{eq:defPrec}\numberthis
\]
for regions \(\Reg_1,\Reg_2 \subset \RealNum^d\) is transitive and anti-symmetric
(in the stronger sense that \(\Reg_1 \prec_\Ww \Reg_2\) and \(\Reg_2 \prec_\Ww
\Reg_1\) imply \(\Reg_1 = \Reg_2 = \emptyset\)). Thus the precursor
relation is a partial order, which is in fact Poincaré covariant. 
\begin{Prop}
  \label{prop:ordCov}
  For any \(\lambda = (x,\Lambda) \in \PoincareGroup\), any wedge \(\Ww\) and
  any sets \(\Reg_1, \Reg_2 \subset\RealNum^{\s+1}\) we have
  \[
    \Reg_2 \prec_\Ww \Reg_1 \Longleftrightarrow 
    \lambda \Reg_2 \prec_{\Lambda \Ww} \lambda \Reg_1.
    \label{eq:ordCov}\numberthis
  \]
  \proof Follows from the elementary computation
  \begin{align*}
    \Reg_2 \prec_\Ww \Reg_1 & \Longleftrightarrow
    \Reg_1-\Reg_2 \subset \Ww_\C 
    \Longleftrightarrow
    \Lambda\Reg_1 - \Lambda\Reg_2 \subset \Lambda\Ww_\C
    \\&\Longleftrightarrow
    \Lambda\Reg_1 + x - \Lambda\Reg_2 - x \subset (\Lambda\Ww)_c
    \Longleftrightarrow
    \lambda\Reg_2 \prec_{\Lambda\Ww} \lambda\Reg_1. \qedhere
  \end{align*}
\end{Prop}

It is clear that the causal complement \(\Ww'\) of any wedge region
\(\Ww\) is also a wedge-region, and that \((\Ww_{\C})'= (\Ww')_{\C}\). 
We say that \(\Ww'\) is the \term{complementary} wedge to \(\Ww\).
More generally we will say that a wedge \(\Ww^\perp\) is {\em opposite} to a given wedge
\(\Ww\) if \(\Ww^\perp\) can be translated into the complement of \(\Ww\), i.e.\ if for
some \(x \in \RealNum^{d}\) we have \(\Ww^\perp + x \subset \Ww'\).
Lastly we will see that the construction of scattering states is most convenient for
the geometrical situation of a given wedge whose edge is
parallel to the time-zero hyperplane.
This is equivalent to \(\Ww = R\WwR + x\) for \(x \in \RealNum^d\) and
some spatial rotation \(R \in \SO(s) \subset \LorentzGroup_+^\uparrow\),
and we will call such wedges~\(\Ww\) \term{upright} or \term{non-tilted}.
This is relevant as for upright~\(\Ww\) the restriction of \(\prec_\Ww\) to
certain hyperplanes behaves almost like a total relation, which will be
helpful for establishing the Fock structure of scattering states
in \Cref{sec:proofhrw}.
\begin{Lem}[``quasi-totality'' of \(\prec_\Ww\) for velocity supports]
  \label{lem:uprightW}
  Let \(\Ww\) be an upright wedge and let \(\VS_k, \VS_k' \subset \RealNum^{\s+1}\), (\(k=1, 2\)), be
  sets of the form (``velocity supports'')
  \[
    \VS_k = \{1\} \times \vVS_k, \quad \vVS_k \subset
    \RealNum^{\s}, \quad (\text{similarly for }\VS_k')
    \label{eq:defVeloLike}\numberthis
  \]
  satisfying
  \[
    \VS_2 \prec_\Ww \VS_1, \quad\text{and}\quad
    \VS_2'\prec_\Ww \VS_1'.
    \label{eq:assumOrd} \numberthis
  \]
  Then necessarily at least one of the two relations
  \[
    \VS_2' \prec_\Ww \VS_1, \quad\text{or}\quad
    \VS_2 \prec_\Ww \VS_1' 
    \label{eq:autoOrd} \numberthis
  \]
  must be satisfied as well.

   \proof
   Let \(\Lambda\) be s.t.\ \(\Lambda\Ww = \WwR\) and note that as
  \(\Ww\) is upright we can choose \(\Lambda\) as a spatial rotation. We
  obtain by \Cref{prop:ordCov} that
  \[
    \Lambda \VS_2 \prec_{\WwR} \Lambda \VS_1, 
    \text{ and }
    \Lambda \VS_2' \prec_{\WwR} \Lambda \VS_1'. 
  \]
  Due to the choice as spatial
  rotation, 
  the sets \(\bar \VS_k:= \Lambda \VS_k \)
  are still of the form~\eqref{eq:defVeloLike}, and analogously for \(\bar \VS_k'\).
  Dropping bars, the two assumptions~\eqref{eq:assumOrd} for \(\Ww = \WwR\)
  translate to inequalities
  \[
    \vec e_1 \cdot (\vec g_1 - \vec g_2) > 0
    \text{ and }
    \vec e_1 \cdot (\vec g_1' - \vec g_2') > 0
    \; \forall \; \vec g_k \in \vVS_k, \;
  \vec g_k' \in \vVS_k', \;(k=1,2),
  \label{eq:ineqAss}\numberthis
  \]
  where \(\vec e_1 \in \RealNum^\s\) denotes the spatial unit-vector in \(1\)-direction.
  Assuming that \(\VS_2' \prec_{\WwR} \VS_1\) is false, 
  there must be \(g_2'^* \in \VS_2'\), \(g_1^* \in \VS_1\) forming an ordering
  ``obstruction''. 
  Namely,
  \begin{align*}
    \neg (\VS_2' \prec_{\WwR} \VS_1) 
    & \Longleftrightarrow 
    \neg (\forall \; \vec g_1 \in \vVS_1\;
       \forall \; \vec g_2' \in \vVS_2':\; 
     \vec e_1 \cdot (\vec g_1 - \vec g_2') > 0)
     \\
& \Longleftrightarrow 
\exists \; \vec g_1^* \in \vVS_1\;
 \exists \; \vec g_2'^* \in \vVS_2':\; 
\vec e_1 \cdot (\vec g_1^* - \vec g_2'^*) \leq 0.
  \label{eq:ineqNAss}\numberthis
  \end{align*}
  For any given \(\vec g_2 \in \vVS_2\) 
  and \(\vec g_1' \in \vVS_1'\) we can now estimate by transitivity
  \[
    \vec e_1 \cdot (\vec g_1' - \vec g_2)
    =
    \vec e_1 \cdot (\vec g_1' - \vec g_2'^*) +
    \vec e_1 \cdot (\vec g_2'^* - \vec g_1^*) +
    \vec e_1 \cdot (\vec g_1^* - \vec g_2) > 0,
  \]
  where we used that the first and last term on the right are strictly
  positive for any \(\vec g_2 \in \vVS_2\) 
  and \(\vec g_1' \in \vVS_1'\) as particular instances of
  \eqref{eq:ineqAss} and the middle term is non-negative due to \eqref{eq:ineqNAss}.
  By definition this implies \(\VS_2 \prec_{\WwR} \VS_1'\).
  \qed \end{Lem}

Finally, let us remark that given any Haag-Kastler net of von Neumann algebras
\(\Reg \mapsto \Alg(\Reg)\) defined for open bounded regions \(\Reg\subset
\RealNum^d\), there exists a canonical associated wedge-local net.
On the other hand starting from a wedge-local net the question of existence or
non-existence of local observables can be highly non-trivial, as explained in
the recent review of Lechner~\cite{Le15}.
While previously the existence of suitable localized operators was always
regarded as essential for going beyond two-particle scattering states,
cf.~\cite{BBS01} or \cite{Le06t} Section~6,
we will see in the following that scattering theory in 
most wedge-local models can be studied  in reasonable generality
without any reference to local observables.

  \section{Construction of Scattering States}
 \label{sec:scatt}
\subsection{Swapping Relations for Opposite Wedge Algebras}
\label{sec:swap}

At the core of our subsequent arguments to establish convergence and Fock
structure of scattering states will be certain swapping identities, such as
\eqref{eq:introSwap}. Due to the
mass gap assumption and with our desired application to the construction of
scattering states it will be in fact sufficient to impose \eqref{eq:introSwap} only
after projection to the one-particle subspace.
\begin{Def}[swapping symmetry of single-particle states] \label{def:swap}
  We say that a single-particle vector \(\Psi_1 \in E_m\HilbertSpace\) of
  mass \(m > 0\) is {\em swappable} with respect to a given wedge~\(\Ww\) if there exist
  operators  \(A \in \Alg(\Ww)\), \(A^\perp \in \Alg(\Ww^\perp)\) localized in
  \(\Ww\) and an opposite wedge~\(\Ww^\perp = \Ww' + x\)
  such that \[
      \Psi_1 = E_m A \Omega = E_m A^\perp \Omega. \label{eq:swap}\numberthis
    \]
\end{Def}

As a matter of fact, swapping relations~\eqref{eq:introSwap}, \eqref{eq:swap} can be obtained as a consequence wedge
duality~\eqref{eq:HK2s}, which is a basic and
well-established structural property in quantum field theory.

\begin{Lem}[D.\ Buchholz, private communications (2017)]
    \label{lem:bu}
  In the vacuum representation of a net 
  \(\Alg\) satisfying wedge-duality~\eqref{eq:HK2s} there exist
  nontrivial vectors satisfying the swapping relation, 
  i.e.\ for suitable \(A \in
  \Alg(\Ww)\) and \(A^\perp \in \Alg(\Ww')\) we have
  \[
    \Psi =  A\Omega = A^\perp\Omega. \label{eq:introSwapen}\numberthis
  \]
  Moreover under the same assumptions, the subspaces \(\HilbertSpace^\Ww \subset
  \HilbertSpace\) of swappable vectors~\(\Psi\) associated to each wedge \(\Ww\) are dense.
\end{Lem}
\noindent
Let us briefly motivate and introduce the mathematical preliminaries necessary
for the proof of \Cref{lem:bu} by noting that in Wightman quantum field
theories there is a natural relation between oppositely localizable operators
provided by the TCP-Operator. In the general operator-algebraic framework, a
similar mapping is accessible abstractly by invoking Tomita-Takesaki theory. In
the following we will only give the basic definitions and refer to \cite[Sec.\
  2.5]{BR87} for proofs and further details.

The Tomita-Takesaki construction starts from the observation that the
adjoint operation \(A \mapsto A^*\) can encode non-trivial
information about the structure of a general
von Neumann algebra~\(\MAlg\). In particular if \(\MAlg\) has a cyclic and separating
vector \(\Omega\) one may define the closable {\em Tomita operator} 
\(
S : \MAlg\Omega \longrightarrow \MAlg\Omega\) by setting
\[
      SA\Omega := A^*\Omega, \quad
  S = J\Delta^{1/2},
    \label{eq:deftom}\numberthis
\]
where the positive self-adjoint {\em modular operator} \(\Delta\) and the
anti-unitary {\em modular conjugation} \(J\) are obtained by 
polar decomposition.
The Tomita-Takesaki theorem
\cite[Thm.\ 2.5.14]{BR87} states that
\[
  J\MAlg J = \MAlg', \quad \text{and} \quad \Delta^{\Ii \tau} \MAlg \Delta^{-\Ii
  \tau} = \MAlg, \; (\tau \in \RealNum).
    \label{eq:ttt}\numberthis
\]

In our case we take \(\MAlg = \Alg(\Ww)\), so that the modular objects
\(S_\Ww\), \(J_\Ww\) and \(\Delta_\Ww\) will depend on the wedge \(\Ww\).
It is clear that \(S_\Ww \Omega = S_\Ww \Id \Omega 
  = \Omega = S_{\Ww'}\Omega\) and one has further \cite[Prop.\ 2.5.11]{BR87}
    \[ \Delta_\Ww \Omega =\Omega, \; J_\Ww \Omega = \Omega.
      \label{eq:ttvac}\numberthis
  \]
In this notation the basic idea for the proof of \Cref{lem:bu} is that for
given self-adjoint \(A = A^* \in \Alg(\Ww)\), \(S_\Ww\) acts trivially on
\(A\Omega\) so that \eqref{eq:deftom} and \eqref{eq:ttvac} yield
 \[
  A\Omega = A^*\Omega = S_\Ww A\Omega = J_\Ww \Delta^{1/2}_\Ww A\Omega =
  \underbrace{
  J_\Ww\Delta^{1/2}_\Ww A\Delta^{-1/2}_\Ww J_\Ww }_{=:A^\perp }\Omega. 
  \label{eq:swapproof}\numberthis
\]
Here a candidate for~\(A^\perp\) can be extracted 
up to domain questions, assuming that~\eqref{eq:ttt} applies here also for imaginary
\(\tau = - \Ii /2 \). 
Technically it remains to show that
\(\bar A := \Delta^{1/2}_\Ww  A \Delta^{-1/2}_\Ww \) 
makes sense as {\em bounded}\/ operator which is in \(\Alg(\Ww)\),
because then \eqref{eq:ttt} and wedge-duality give 
\(A^\perp :=  J_\Ww  \bar A  J_\Ww \in \Alg(\Ww)' = \Alg(\Ww')\).

Before proceeding to the proof of \Cref{lem:bu}, we should point out that the
swapping relation~\eqref{eq:introSwapen} is established exactly for 
``touching'' wedges~\(\Ww^\perp = \Ww'\). Then \(\Ww^\perp \cap \Ww\)
is empty, so that in this case \eqref{eq:introSwapen} is non-trivial also for local theories.
Let us recall that wedge duality~\eqref{eq:HK2s} can be proven in the Wightman
framework via the Bisognano-Wichmann property~\cite{BW75} 
\[ 
  U(\Lambda_\Ww(2\pi\tau)) = \Delta_{\Ww}^{-\Ii \tau}.
  \label{eq:BW}\numberthis
\]
Here \(U\) denotes the unitary implementation of 
Poincaré symmetry, and \(\Lambda_\Ww(2\pi\tau)\) is the
oriented one-parameter group of boosts preserving the given wedge~\(\Ww\).
In the setting of standard subspaces,
wedge-duality \eqref{eq:HK2s} and \eqref{eq:BW} have
recently been established for any finite or infinite multiplets of massive
or massless scalar irreducible unitary representations of the Poincaré group
\cite{Mo17}.

\begin{proof}[Proof of \Cref{lem:bu}] We follow the argument of Buchholz
  \cite{Bu17pc}. As we keep \(\Ww\) fixed, we drop wedge indices on the modular
  objects. To establish existence we consider vectors of the form~\(\Psi = A \Omega\)
  with \(A^* = A \in\Alg(\Ww)\).  
  Rigorous control over \eqref{eq:swapproof} is then obtained by passing to 
  operators~\(A_\delta\), (\(\delta > 0\)), which are ``regularized'' with
  respect to the adjoint action of the modular group by setting
  \[
    A_\delta :=
    \int \frac{ \DInt \tau}{ \sqrt{2\pi\delta}} \;
    \Ee^{-\frac{\tau^2} {2\delta}}
    \;
     \Delta^{\Ii \T} A \Delta^{-\Ii \T}.
     \label{eq:defAdelta}\numberthis
  \]
  From the Tomita-Takesaki theorem~\eqref{eq:ttt} we see that the integrand is
  pointwise in \(\Alg(\Ww)\) so that \(A_\delta \in \Alg(\Ww)\) as 
  wedge-algebras are weakly closed.
  Secondly we obtain from strong continuity of \(\Delta^{\Ii \tau}\) that
  \(A_\delta\rightharpoonup A\) in the weak operator topology, so that by modular
  invariance of \(\Omega\) we have \( A_\delta\Omega \to A\Omega\) in norm as
  \(\delta \to 0\).
  Further due to \eqref{eq:defAdelta} the adjoint action of the modular group on
  \(A_\delta\)
  may be computed explicitly as
  \[
    \Delta^{\I t} A_\delta \Delta^{-\I t} 
    =  
    \int \frac{\DInt \tau}{\sqrt{2\pi\delta}} \;
    \Ee^{-\frac{\tau^2} {2\delta}}
    \Delta^{\Ii \T + \Ii t} A \Delta^{-\Ii \T - \Ii t}
  =  
    \int \frac{\DInt \tau'}{\sqrt{2\pi\delta}} \;
    \Ee^{-\frac{(\tau'-t)^2} {2\delta}}
    \Delta^{\Ii \tau'} A \Delta^{-\Ii \tau'}.
     \label{eq:adjAct}\numberthis
  \]

  Returning to \eqref{eq:swapproof} we now define 
  \(\bar A_\delta : = \Delta^{1/2} A_\delta \Delta^{-1/2}\)
  as a quadratic form on a suitable domain.  It will be convenient to restrict
  to \( \Dom_\omega(\Delta^\pm)\ := \{ E_\Delta([k,K])\Psi, \Psi \in
    \HilbertSpace, 0 < k < K \}\), which is dense in \(\HilbertSpace\) by
  spectral calculus.  For \(\Psi_1, \Psi_2 \in\Dom_\omega(\Delta^\pm)\) the
  function \(t \longmapsto \langle \Psi_1, \Delta^{\I t} A_\delta \Delta^{-\I
      t} \Psi_2 \rangle \) is entire analytic.
  It further coincides for \(t \in \RealNum\) with the entire function
  defined by the right hand side of \eqref{eq:adjAct}. By analyticity these
  two entire functions coincide for all \(t \in \ComplexNum\) so that
  \begin{align*} \langle \Psi_1, \Delta^{1/2}
    A_\delta \Delta^{-1/2} \Psi_2 \rangle
    & =
   \int \frac{ \DInt \tau}{ \sqrt{2\pi\delta}} \;
    \Ee^{-\frac{\tau^2} {2\delta}}
    \;
    \langle \Psi_1, \Delta^{\Ii \T + 1/2} A \Delta^{-\Ii \T - 1/2} \Psi_2 \rangle
    \label{eq:cint1}\numberthis
     \\&
     = 
 \int \frac{ \DInt \tau'}{ \sqrt{2\pi\delta}} \;
 \Ee^{\left.-\frac{(\tau'+  \Ii/ 2)^2} {2\delta}\right.}
    \;
    \langle \Psi_1, 
    \Delta^{\Ii \T'} A \Delta^{-\Ii \T'}
    \Psi_2 \rangle. \label{eq:cint2}\numberthis
  \end{align*}

  From~\eqref{eq:cint2} we see firstly that \eqref{eq:cint1} in fact defines a bounded
  bilinear form, so that \(\bar A_\delta\) extends to
  a bounded operator on all of \(\HilbertSpace\), and secondly that \(\bar
    A_\delta \in \Alg(\Ww)\) by repeating the argument below
    \eqref{eq:defAdelta}.
  Thus the swapping partner may be obtained as in \eqref{eq:swapproof} by
  setting \(A^\perp_\delta:= J\bar A_\delta  J\), and noting that~\(A^\perp_\delta
    \in \Alg(\Ww')\) due to \eqref{eq:ttt} and wedge duality \eqref{eq:HK2s}. 

  To establish density of swappable vectors let \(\Psi \in \HilbertSpace \)
  and \(\epsilon > 0\).
  By cyclicity of \(\Omega\) there exists \(A \in \Alg(\Ww)\) such that 
  \(\norm{\Psi - A\Omega} \leq \epsilon/2\). We may then decompose \(A = \frac 1
  2(A + A^*) + \frac 1 2 (A-A^*) =: A_1 + \Ii A_2\) such that the above argument
  applies to \(A_k \Omega\), \(k=1,2\), and the swapping partner of
  \(A_\delta\) is then given by \(A_\delta^\perp := (A_1)_\delta^\perp +
  \Ii (A_2)_\delta^\perp\).
  Choosing \(\delta > 0\) sufficiently small yields
  \(\norm{\Psi - A_{\delta} \Omega} \leq \norm{\Psi - A\Omega} +
    \norm{A_1\Omega - (A_1)_\delta\Omega} + \norm{A_2\Omega -
  (A_2)_\delta\Omega} \leq  \epsilon\) so that we obtain density.
\end{proof}

\begin{Cor}
  Assuming \eqref{eq:HK2s}, single-particle vectors
 satisfying the swapping relation~\eqref{eq:swap} w.r.t.\ any given 
  wedge \(\Ww\) are dense in the single-particle space~\(\HilbertSpace_1 :=
  E_m\HilbertSpace\).
\end{Cor}

We note that for space-time
dimension \(d \geq 2+1\) the dense sets of swappable vectors constructed in
\Cref{lem:bu} in general have a non-trivial dependence on~\(\Ww\).
Interestingly certain wedge-local models also admit a dense subspace of vectors 
which are swappable in the sense of \Cref{def:swap} for all wedges
simultaneously, as can be seen from the results of \cite{BLS10} in the class of
deformed local theories.

A simple and immediate consequence of the swapping relation is the consistency
of our definition of scattering states~\eqref{eq:scattApprox} with previous
discussions of two-particle scattering in wedge-local models \cite{GL07, BS08},
where the physically obvious opposite-localization prescription~\(\Psi^+ :=
  \lim_{\T \to \infty} \BB_{\T}(f) \BB_{\T}^\perp(f^\perp)\Omega \) has been
  used. 
With the swapping relation as main technical tool at hand, we may in fact swap 
\[\Psi^+ = \lim_{\T \to \infty} \BB_{\T}(f) \BB_{\T}^\perp(f^\perp)\Omega
 =  \lim_{\T \to \infty} \BB_{\T}(f) \bar\BB_{\T}(f^\perp)\Omega,
 \label{eq:ans2}\numberthis
\] 
where \(\bar\BB_{\T}(f^\perp) \) is defined in terms of 
\(\bar A \in \Alg(\Ww)\) with \(\bar A \Omega = A^\perp \Omega\).
The new prescription from \eqref{eq:ans2} with all operators localized in the
same wedge~\(\Ww\) now generalizes to \(N\)-particle scattering theory, as 
will be seen in the next section.

\subsection{Wedge-local Haag-Ruelle Theorem} 
\label{sec:proofhrw}

As comparison dynamics for the construction of scattering states we may restrict
to {\em regular positive-energy}\/ Klein-Gordon solutions~\(f_k\), which are
of the form
\begin{align*}
  f_k(t, \vec x) &= \int \frac{\DInt[s] k}{(2\pi)^\s}
  \Ee^{\Ii \vec k \cdot \vec x - \Ii \omega_m(\vec k) t} \;
  \tilde f_k(\vec k), 
  \\
  \omega_m(\vec k) &:= \sqrt{\vec k^2 + m^2},\;
  \tilde f_k \in \ContinuousFuncs^\infty_c(\RealNum^{\s}).
  \label{eq:defKG} \numberthis
\end{align*}

\begin{Def}[Haag-Ruelle creation operator approximants]
  For \(A \in \Alg(\Ww)\), \(\chi \in \SchwartzSpace(\RealNum^{\s+1})\),
  and \(f\) a regular positive-energy Klein-Gordon solution we set for \(\T \in
  \RealNum\)
  \begin{align*} 
    \Bb &:= A(\chi) = \int \DInt[\s+1] x \; \chi(x) \alpha_x(A),
      \label{eq:defCreat2}\numberthis \\
    \BB_\T(f) &:= \int \DInt[\s] x\; f(\T, \vec x) \alpha_{(\T, \vec x)}(A).
      \label{eq:defHR2}\numberthis
  \end{align*}
\end{Def}
For our main result (\Cref{thm:hrw})
and in the following we will always assume \(\chi\) 
to be chosen as in \Cref{lem:HRW} below, in accordance with the mass gap~\eqref{eq:HK6}.
The restrictions on propagation of wave packets mentioned in the introduction
are made precise using the precursor relation~\eqref{eq:defPrec},
to constrain the {\term{velocity supports}}
\[
  \VS_{f_k} := \{
    (1 , \vec k/\omega_m(\vec k)),\;
    \vec k \in \support \tilde f_k
  \}.
  \label{eq:defVelo}\numberthis
\]
Basic intuition for handling localizations of creation-operator approximants
comes from the fact that regular~\(f_k\) are rapidly decreasing outside the
cone~\(\Upsilon_{f_k}^\delta := \RealNum \VS_{f_k}^\delta\) generated by any
\(\delta\)-neighbourhood \(\VS_{f_k}^\delta \supset \VS_{f_k}\),
as seen from standard non-stationary phase estimates\footnote{The velocity support estimates for
  regular Klein-Gordon solutions are due to Ruelle \cite{Ru62}, for details see
e.g.\ \cite[Thm.\ 5.3]{ArQFT99}. Via such estimates,  disjointness
  \(\VS_k \cap \VS_j = \emptyset\), (\(k \not = j\))  is sufficient for local QFT
  to control equal-time commutators \cite{Hep65}, and to some limited extent
also non-equal time-commutators~\cite{MD16}.}.

\begin{Thm}
  \label{thm:hrw}
  Fix a wedge \(\Ww\) and let \(\Psi_k \in \HilbertSpace_1\) \((1\leq k \leq
  n)\) be single-particle vectors isolated from the remaining energy-momentum
  spectrum which satisfy the swapping relation
  \(\Psi_k = E_m A_k \Omega = E_m A_k^\perp\Omega\), \(A_k \in \Alg(\Ww)\), \(A_k^\perp
    \in \Alg(\Ww^\perp)\).
  \begin{enumerate}[(i)]
  \item \label{it:hrwconv} For any family of regular positive-energy
    Klein-Gordon solutions \(f_k\) satisfying
    \begin{align*}
      &
      \VS_{f_n} \prec_\Ww \VS_{f_{n-1}} \prec_\Ww \ldots \prec_\Ww \VS_{f_1},
      \label{eq:ordPrescN}\numberthis
      \\
      \Psi_\T &:= \BB_{1\T}(f_1) \BB_{2\T}(f_2) \ldots \BB_{n\T}(f_n)\Omega \quad
      (\T \in \RealNum)
      \label{eq:defOut}\numberthis
    \end{align*}
    converges in norm for \(\T\to \infty\). 

  \item \label{it:hrwfock}
    Let \(\Psi^+ := \lim_{\T \to \infty} \Psi_\T\),
    \(\Psi'^+ := \lim_{\T \to \infty} \Psi_\T'\) be scattering states as
    in \itref{it:hrwconv}, constructed
    from operators localizable with respect to the same wedge \(\Ww\).
    Then for upright~\(\Ww\) their
    scalar products can be computed using the Fock
    prescription
    \[
    \left\langle \Psi^+, \Psi'^+\right \rangle
    = \delta_{nn'} \prod_{k=1}^n \left\langle \BB_{k\T}(f_k)\Omega, \BB_{k\T}'(f_k')\Omega \right\rangle,
    \label{eq:fock}\numberthis
    \]
    where the right-hand side is independent of \(\T\).
  \end{enumerate}
  Analogous statements hold for the convergence and Fock structure of any two incoming
  scattering states (\(\T \to - \infty\)) defined using the reversed 
  ordering of wave packets
  \[
    \VS_{f_n} \succ_\Ww \VS_{f_{n-1}} \succ_\Ww \ldots \succ_\Ww \VS_{f_1}.
    \label{eq:ordPrescNSucc}\numberthis
  \] 
\end{Thm}

We should point out that the ordering prescription~\eqref{eq:ordPrescN} 
is not new. Such relations are well known in the form-factor
programme and related constructive work, see e.g.\ \cite[p.\
8]{Smi92} and references therein, or \cite[Sec.\ 6]{Le06t}. 
However in contrast to the results from \cite{BBS01, Le06t}, we note that our
arguments require neither the existence of local observables, nor temperateness
of suitable polarization-free generators.

\begin{Lem}[Haag-Ruelle Lemma, wedge-local version] \label{lem:HRW}
  Let \(A \in \Alg(\Ww)\) and 
  \(K \subset K' \subset H_m\) be compact subsets of the mass shell, such that
  \(K\) can be separated from \(H_m \setminus K'\) by a smooth function. 
  Then there exists a suitable \(\chi \in \SchwartzSpace(\RealNum^{\s+1})\) 
  (with \(\hat \chi\) supported in a sufficiently small neighbourhood of the mass shell as dictated
  by the mass gaps~\eqref{eq:HK6})
  such that \(B := A(\chi)\) satisfies
  \begin{enumerate}[(i)] 
    \item \(B \Omega \in E(K' )\HilbertSpace\subset E(H_m)\HilbertSpace\),
      \label{it:one}
  \item \(E(K) B \Omega = E(K)A \Omega\),
  \item \(B^*\Omega = 0\), \label{it:annihil}
  \item \(B^* \Psi_1 = E_\Omega B^* \Psi_1 \) 
    for all \(\Psi_1 \in E(K' \cap H_m)\HilbertSpace\),
    where \(E_\Omega := |\Omega \rangle\langle \Omega|\). 
    \label{it:clu}
  \item \label{it:awl} \(B\) is {\em almost
    wedge-local} (w.r.t.\ \(\Ww\)), i.e.\ for any \(r > 0 \) there exists \(B_r \in
    \Alg(\Ww + \DoubleCone_r)\) 
    so that for any \(N \in \NaturalNum\) we have for a suitable \(C_N >0\) that
    \[ \norm{B - B_r} \leq \frac{C_N}{1+r^N}. \label{eq:awlDecay} \numberthis \]
    Here 
    \(\DoubleCone_r := \{ x \in \RealNum^{\s+1}: \abss{x^0}+\abss{\vec x}<r\}\)
    denotes the double cone of radius \(r\).  
  \end{enumerate}
\end{Lem}
\Cref{lem:HRW} has a well-known counterpart in strictly local theories
\cite{Ha58} \cite{Ru62}, which allows us to skip the proof. 
In particular the main spectral statements \itref{it:one},
\itref{it:annihil} and \itref{it:clu} may be
understood by noting that the smearing operation \(\Bb:=\Aa(\chi)\) restricts the
Arveson spectrum\footnote{See e.g.\ \cite{Arv80}  or \cite{BDN15}, Sec.\ 3.}
\({\rm Sp}_\Bb \alpha \subset \support \hat \chi\).  The only modification
appears in \itref{it:awl}, where
the statement of the lemma needs to be adapted for the wedge-local case.
From \Cref{lem:HRW} we immediately obtain analogous properties for the
creation-operator approximants~\(\BB_\T(f)\) defined by the standard LSZ prescription~\eqref{eq:defHR2}.
\begin{Prop}[elementary properties of \(\Bb\) and \(\BB_\T\)] 
  \label{prop:otherProp}
  ~
  \begin{enumerate}[(i)]
    \item 
      \(\BB_\T(f)\Omega = \tilde f(\vec P) \Bb \Omega\)
      for all \(\T \in \RealNum\).
    \item If \(A\Omega = A^\perp\Omega\), the corresponding Haag-Ruelle
      operators satisfy \(\BB_\T(f)\Omega = \BB^\perp_\T(f)\Omega\).
    \item \(\partial_\T \BB_\T(f) \Omega = 0\). \label{it:diffBT}
    \item \(\norm{\BB_\T(f)} 
      \leq C (1+ \abs{\T}^{\s/2})\).
      \label{it:nest}
    \item  \label{it:diff}
      \(\partial_\T \BB_\T(f)\) exists in norm and
      \(\norm{\partial_\T \BB_\T(f)} \leq C' (1+ \abs{\T}^{\s/2})\).
    \item \( \BB_{1\T}(f_1)^* \, \BB_{2\T}(f_2) \Omega 
      = E_\Omega \BB_{1\T}(f_1)^* \, \BB_{2\T}(f_2) \Omega \),
      \label{it:cluBT}
      independently of velocity supports and operators possibly associated to
      different wedges~\(\Ww_1\), \(\Ww_2\), 
      where \(E_\Omega := |\Omega \rangle \langle \Omega |\).
  \end{enumerate}
\end{Prop}
Having adapted the statements of \Cref{lem:HRW} and \Cref{prop:otherProp}
as required by wedge-locality, we will skip the proofs which carry over
literally  from standard Haag-Ruelle theory (up to weakened localization)
and refer to \cite[Sec.~5]{ArQFT99} or \cite{MD16} for further details. 
Still the most significant consequence of wedge-locality for Haag-Ruelle
theory is contained in the following localization and commutator estimates, whose
proofs will be sketched for the convenience of the reader in \Cref{app:proofs}.

  \begin{Lem}
    \label{lem:wedgeLocalBT}
    Let \(\Aa \in \Alg(\Ww)\). 
    For any \(\T \in \RealNum\) and \(\delta > 0\) the corresponding 
    \(\BB_\T :=  \BB_{\T}(f)\) can be approximated 
    by \(\BB_{\T}^{(\delta)} \in \Alg(\T \VS_f + \DoubleCone_{\delta\abs{\T}} + \Ww) \),
  \((\delta > 0)\), such that
    for any \(N \in \NaturalNum\)
    \[
      \norm{\BB_{\T}^{(\delta)} - \BB_\T} \leq \frac{ C_N^\delta}{1+\abs{\T}^N},
      \label{eq:asymptDec}\numberthis
    \]
    where the constants \(C_N^\delta\) depend 
    on \(f\), \(\Aa\) and \(\chi\), but are independent of \(\T\). 

    For later use in \Cref{sec:genW} we note that 
      analogous 
      approximants~\(\bar \BB_{\T}^{(\delta)}\) exist if \(f\) is replaced by
      the pointwise product~\(\bar f := f h\)  with a polynomially bounded
      measurable function~\(h:\RealNum^d \to \ComplexNum\). 
  \end{Lem}

\begin{Cor}[commutators with ordered velocity support] 
  \label{cor:commEst}
  Let \(B\), \(B^\perp\) be 
  as in \Cref{lem:HRW} for a pair of opposite wedges \(\Ww, \Ww^\perp\), respectively,  and
  let \(f, f^\perp\) be 
  ordered by \(\VS_{f^\perp} \prec_{\Ww} \VS_f\).
  Then for any \(\T > 0\),
  \[
    \norm{\left[\BB_\T^{\perp}(f^\perp), \BB_\T(f) \right] }
    \leq \frac{C_N}{1+\abs{\T}^N},
    \label{eq:commEst}\numberthis
  \]
  where \(C_N\) depend on operators and smearing functions as in
  \Cref{lem:wedgeLocalBT}.
  For \(\T < 0\) estimate~\eqref{eq:commEst} holds under the reversed ordering
  assumption \(\VS_f \prec_{\Ww} \VS_{f^\perp}\).
  The commutator estimate extends to the cases that one or both of the operators in
  \eqref{eq:commEst} are replaced by their adjoints or \(\T\)-derivatives.
\end{Cor}

\begin{proof}[Proof of \Cref{thm:hrw}] 
  {\bf Ad (i)}
  Setting
  \(\Psi_\T : = \Psi_\T^{(n)} :=  \BB_{1\T}(f_1) \BB_{2\T}(f_2) \ldots \BB_{n\T}(f_n)\Omega\) we
  would like to establish convergence for \(\T \to \infty\).
  Due to \Cref{prop:otherProp} \itref{it:diff} and \itref{it:nest}, Cook's method is applicable and
  we can write for \(0 < \T_1 < \T_2\)
  \[
    \norm{\Psi_{\T_2} - \Psi_{\T_1}}  = \norm{\int_{\T_1}^{\T_2} \DInt \T \;
    \partial_\T \Psi_\T} \leq \int_{\T_1}^{\T_2} \DInt \T \; \norm{
    \partial_\T \Psi_\T}.
       \label{eq:cook}\numberthis
  \]
  Convergence will follow from the rapid decay estimate 
  \(\norm{\partial_\T \Psi_\T} \leq C_N \T^{-N}\) for \(\T > 0\).
  
  The latter is obtained by induction with respect to the number of
  particles \(n\), with starting case~\(n = 1\) given by
  \(\partial_\T \Psi_\T = 0\) as seen in \Cref{prop:otherProp} \itref{it:diffBT}.
  For the induction step 
  we write
  \begin{align*}
    \partial_\T \Psi_\T 
    &= \partial_\T (\BB_{1\T}(f_1) \BB_{2\T}(f_2) \ldots
      \BB_{n-1\,\T}(f_{n-1})) \; \BB_{n\T}(f_n)\Omega
      \\ &\qquad +  \BB_{1\T}(f_1)  \ldots \BB_{n-1\,\T}(f_{n-1}) \; \partial_\T\BB_{n\T}(f_n)\Omega
      \\
      &= \partial_\T (\BB_{1\T}(f_1) \BB_{2\T}(f_2) \ldots
      \BB_{n-1\,\T}(f_{n-1})) \; \BB_{n\T}^\perp(f_n)\Omega,
      \label{eq:swapped}\numberthis
  \end{align*}
  where we first used \Cref{prop:otherProp} \itref{it:diffBT} to drop the term
  with  derivative operator acting directly on the vacuum and used that the
  swapping relation~\eqref{eq:swap} implies~\(\BB_{n\T}(f_n)\Omega = \BB_{n\T}^\perp(f_n)\Omega\). 
  Now there are oppositely wedge-localized pairs of HR-operators whose
  commutators can be controlled using \Cref{cor:commEst}, and may estimate for \(\T > 0\)
  \begin{align*}
    \norm{\partial_\T \Psi_\T}
    &
      \leq
      \norm{  \BB_{n\T}^\perp(f_n)} \norm{ \partial_\T \Psi_\T^{(n-1)} } 
      \\&\qquad \qquad
    + 
     \norm{\left[ \partial_\T \BB_{1\T}(f_1) \ldots
     \BB_{n-1\,\T}(f_{n-1}) , \BB_{n\T}^\perp(f_n) \right]} \norm{\Omega}.
       \label{eq:indConv}\numberthis
  \end{align*}
  Here the first summand is rapidly decreasing 
  for \(\T \to \infty\)
  by the induction assumption and 
  \Cref{prop:otherProp}~\itref{it:nest}. The second summand can be 
  generously bounded from above by expanding the derivative and commutator
  as
  \begin{align*}
    \sum_{k=1}^{n-1} \sum_{j=1}^{n-1}
      \BB_{1\T}(f_1) \ldots (\partial_\T\BB_{k\T}(f_k)) \ldots
      \left[\BB_{j\T}(f_{j}) , \BB_{n\T}^\perp(f_n) \right] \ldots
      \BB_{n-1\,\T}(f_{n-1}).
    \label{eq:expComm}\numberthis
  \end{align*}
  Estimating the corresponding operator norm in \eqref{eq:indConv} by expanding in terms
  of
  \( \norm{\BB_{k\T}} \leq C_k (1+\abs{\T}^{s/2})\),
  \( \norm{\partial_\T \BB_{k\T}} \leq C_k' (1+\abs{\T}^{s/2})\), 
  \( \norm{ \left[\BB_{j\T}(f_{j}) , \BB_{n\T}^\perp(f_n) \right]} \leq C_N
  (1+\T)^{-N}\), and 
  \( \norm{ \left[\partial_\T\BB_{k\T}(f_{k}) , \BB_{n\T}^\perp(f_n) \right]} \leq C_N
  (1+\T)^{-N}\) yields an overall rapid decay.
  Here we used that \Cref{cor:commEst} applies due to transitivity of the
  precursor ordering. Together we obtain that \eqref{eq:indConv} decays faster
  than any polynomial, and thus convergence of outgoing scattering states
  follows from \eqref{eq:cook}.
  The existence of incoming states follows analogously for opposite operator ordering.

  {\bf Ad (ii)} 
  Letting \(\Psi^+ := \lim_{\T\to\infty} \BB_{1\T}(f_1) \ldots \BB_{n\T}(f_n) \Omega\)
  and another scattering state 
  \(\Psi'^+ := \lim_{\T\to\infty} \BB_{1\T}'(f_1') \ldots \BB_{n'\T}'(f_{n'}') \Omega\) 
  defined with respect to the same wedge~\(\Ww\), we denote the minimum number
  of particles by \(N:=\min(n,n')\).
  We will assume instead of upright~\(\Ww\) only the following weaker
  technical ordering condition: adjacent pairs of velocity supports are
  precursor-comparable from the rear also across the two families, in the
  sense that 
  \[ \forall \, 0\leq j < N: 
   \VS_{f_{n-j}} \prec_\Ww \VS_{f_{n'-j-1}'} \text{ or }
  \VS_{f_{n'-j}'} \prec_\Ww \VS_{f_{n-j-1}}.
   \label{eq:techOrd}\numberthis\] 
   For upright wedges \eqref{eq:techOrd} follows from \Cref{lem:uprightW}, but
   the argument based on \eqref{eq:techOrd} can be also applied for 
   non-upright~\(\Ww\), e.g.\ to compute \(\norm{\Psi^+}^2 = \langle \Psi^+,
   \Psi^+\rangle\).

  The proof of the Fock relation~\eqref{eq:fock} is now by induction on the
  minimum number of particles~\(N\).
   By
  continuity of the scalar product we may write
  \[
    \left\langle \Psi^+, \Psi'^+ \right\rangle 
    =
    \lim_{\T\to \infty}\left\langle\Omega,
      \BB_{n\T}(f_n)^* \ldots \BB_{1\T}(f_1)^*\, \BB_{1\T}'(f_1') \ldots \BB_{n'\T}'(f_{n'}') \Omega 
      \right\rangle.
      \label{eq:fockind}\numberthis
    \]
  For \(N=0\) the Fock identity~\eqref{eq:fock} follows from \(\norm{\Omega}=1\)
  or \Cref{lem:HRW}~\itref{it:annihil}, in the respective cases 
  vacuum-vacuum or for a non-zero number of creation operators.
  Assuming \eqref{eq:fock} holds for \(N\) particles, we now 
  distinguish the two cases 
  \(\VS_{f_{n}} \prec_\Ww \VS_{f_{n'-1}'}\) or \(\VS_{f_{n'}'} \prec_\Ww
    \VS_{f_{n-1}}\), determining on which side of \eqref{eq:fockind} the swapping 
  should be performed. Let us proceed for the case \( \VS_{f_{n'}'} \prec_\Ww
  \VS_{f_{n-1}}\), by swapping
  \begin{align*}
    \left\langle \Psi_\T, \Psi_\T' \right\rangle &=
      \left\langle\Omega,
      \BB_{n\T}(f_n)^* \ldots \BB_{1\T}(f_1)^*\, \BB_{1\T}'(f_1') \ldots
      \BB_{n'\T}'(f_{n'}') \Omega
      \right\rangle
      \\&=
      \left\langle\Omega,
      \BB_{n\T}(f_n)^* \ldots \BB_{1\T}(f_1)^*\, \BB_{1\T}'(f_1') \ldots
      \BB_{n'\T}'^\perp(f_{n'}') \Omega
      \right\rangle
      \\
      &=
      \left\langle\Omega,
        \BB_{n\T}^* \BB_{n'\T}'^\perp \BB_{n-1\,\T}^*
        \ldots \BB_{1\T}^*\, \BB_{1\T}' \ldots
        \BB_{n'-1\,\T}'
        \Omega
      \right\rangle 
      \\& \qquad +
      \left\langle\Omega,
        \BB_{n\T}^* \left[\BB_{n-1\,\T}^*
        \ldots \BB_{1\T}^*\, \BB_{1\T}' \ldots
      \BB_{n'-1\,\T}', \BB_{n'\T}'^\perp \right] \Omega \right\rangle,
  \end{align*}
  where in the last step and below we suppress obvious wave packet dependences.
  Expanding the commutator gives
  \begin{align*}
    \fl
      \sum_{k=1}^{n-1}
      \BB_{n-1\,\T}^* 
      \ldots \left[ \BB_{k\T}^*,  \BB_{n'\T}'^\perp \right]\ldots
      \BB_{1\T}^*  
\BB_{1\T}' \ldots \BB_{n'-1\,\T}'
\\&\qquad+
    \BB_{n-1\,\T}^* \ldots \BB_{1\T}^*\, 
      \sum_{k=1}^{n'-1}\BB_{1\T}' 
      \ldots \left[ \BB_{k\T}',  \BB_{n'\T}'^\perp \right]\ldots
      \BB_{n'-1\,\T}'.
  \end{align*}
  Here \Cref{cor:commEst} applies due to \( \VS_{f_{n'}'} \prec_\Ww
  \VS_{f_{n-1}}\), the assumed orderings~\eqref{eq:ordPrescN} of the
  velocity supports of \(f_k\) and \(f_k'\) within each family, and transitivity of the
  precursor ordering.
  This yields
  \( \norm{ \left[\BB_{k\T}(f_{k})^* , \BB_{n\T}^\perp(f_n) \right]} \leq C_N (1+\T)^{-N}\) and 
  \( \norm{ \left[\BB_{k\T}'(f_{k}) , \BB_{n\T}^\perp(f_n) \right]} \leq C_N (1+\T)^{-N}\), 
  so that
  together with \( \normm{\BB_{j\T}^*} \leq C_j (1+\abs{\T}^{s/2})\) and \(
  \normm{\BB_{j\T}'} \leq C_j' (1+\abs{\T}^{s/2})\)
  from \Cref{prop:otherProp}~\itref{it:nest},
  we can estimate for~\(\T > 0\) 
  \begin{align*}
    \abs{\left\langle \Psi_\T, \Psi_\T' \right\rangle -
      \left\langle  \Omega, \BB_{n\T}^*\BB_{n'\T}'^\perp \BB_{n-1\,\T}^* \ldots \BB_{1\T}^*\,
        \BB_{1\T}' \ldots \BB_{n'-1\,\T}' \Omega 
      \right\rangle} 
     &\leq C_N \T^{-N}.
            \label{eq:cluEst} \numberthis
  \end{align*}

  As \( \lim_{\T \to \infty} \left\langle \Psi_\T, \Psi_\T' \right\rangle \)
  exists by part (i) of this theorem, which was established above,
  \begin{align*}
     \lim_{\T \to \infty}
     \left\langle \Psi_\T, \Psi_\T' \right\rangle 
   &=
     \lim_{\T \to \infty}
     \left\langle (\BB_{n'\T}'^\perp)^*  \BB_{n\T}\Omega, \BB_{n-1\,\T}^* \ldots \BB_{1\T}^*\,
     \BB_{1\T}' \ldots \BB_{n'-1\,\T}' \Omega \right\rangle
   \\&=
     \lim_{\T\to \infty}\left\langle \Omega, (\BB_{n'\T}'^\perp)^*  \BB_{n\T}\Omega
     \right\rangle \left\langle \Omega, \BB_{n-1\,\T}^* \ldots \BB_{1\T}^*\,
     \BB_{1\T}' \ldots \BB_{n'-1\,\T}' \Omega \right\rangle,
  \end{align*}
  where the right hand side was rewritten using the clustering identity 
  from \Cref{prop:otherProp}~\itref{it:cluBT}.
  The existence of the limit on the right-hand side now follows for the
  one-particle matrix element in the first factor from \Cref{prop:otherProp}
  \itref{it:diffBT}, and for the second factor from the
  induction assumption, respectively. By induction we obtain finally the Fock formula~\eqref{eq:fock}.

  For the complementary ordering \( \VS_{f_{n'}'} \prec_\Ww \VS_{f_{n-1}}\) we swap
  instead on the opposite side of \eqref{eq:fockind}, making use of \(\BB_{n\T}(f_n)\Omega =
  \BB_{n\T}^\perp(f_n)\Omega\). Following otherwise the same chain of arguments
  we obtain the limit~\eqref{eq:fock} also in this case.
\end{proof}

To conclude this section let us recall that in dimension~\(1+1\) all wedges
are upright in a trivial sense. In higher dimension the restriction to upright
wedges seems to be unphysical as it singles out a non-Poincaré-covariant
family of localization wedges. We will later see that the uprightness restriction
is of a technical nature arising due to the a priori Lorentz-frame dependent
formulation of Haag-Ruelle theory.
Consequently it can be lifted by passing to a variant of the Haag-Ruelle
creation operator approximants~\eqref{eq:defHR} adapted to the reference frame of a given
(non-upright) operator localization wedge \(\Ww\).

\section{Localization in General Wedges} \label{sec:genW}
The goal of this section is to remove the assumption of localization of
operators in upright wedges from \Cref{thm:hrw}~\itref{it:hrwfock},
as will be needed for a physically satisfactory discussion of the known
Poincaré-covariant wedge-local models (e.g.\ as in \cite{BLS10}).
We recall that these additional considerations are specific to the case of
spatial dimension~\(s>1\).
The following simple example illustrates the causal restrictions in the
non-upright case which invalidate \Cref{lem:uprightW} and allows to visualize
how these are resolved below.

\begin{Rem}[canonical non-upright wedge] \label{ex:nonuprightw}
A non-upright wedge can be obtained by boosting the right wedge~\(\WwR = \{
x \in \RealNum^d, \; \abs{x^0} < x^1 \}\), \(d \geq 3\), in \(x^2\)-direction with
rapidity~\(\beta \in \RealNum\setminus\{0\}\), 
  yielding
\[
\Ww:=  \Lambda_{\beta}^{(2)} \WwR = \{ x \in \RealNum^d, \; \abs{\cosh(\beta) x^0 -
\sinh(\beta)x^2} < x^1\}. \label{eq:nonupright}\numberthis
\]
\end{Rem}

For concreteness we may take \(d = 3\).
The relevant part determining the precursor ordering of velocity supports
\(\VS_1 \prec_\Ww \VS_2 \Longleftrightarrow \VS_2 - \VS_1 \subset \Ww\) is the
restriction of \(\Ww\) to \(\{x^0 = 0\}\).
For the upright case \(\beta = 0\) this restriction is a half plane,
and the opposite ordering \(\VS_2 \prec_\Ww \VS_1 \) corresponds to inclusion in the
complementary open half-plane.
Exactly this special geometrical situation is necessary
for the validity of \Cref{lem:uprightW}. 
Further this means physically that the scattering states constructed in
\Cref{thm:hrw} cover the entire 2-particle velocity space up to a set of
measure zero.\footnote{Underlying this simple picture is of course the intuition of
conventional (e.g.\ bosonic) particle statistics, which may be misleading in the
general wedge-local setting as illustrated by recent examples of Longo, Tanimoto and
Ueda \cite{LTU17}.}

However for the non-upright case~\(\beta \not = 0\) the restriction of
\(\Ww\) to \(\{x^0 = 0\}\) yields merely a cone \( C :=  \{ \vec x  \in
\RealNum^{d-1},  \abs{\sinh(\beta)x^2} < x^1\} \).
Hence there is a non-trivial region of the two-particle velocity space
which cannot be decomposed into ordered configurations.
For example we may take the corresponding velocity supports concentrated
in sufficiently small neighbourhoods of points \(v_1 \in \VS_1\), \(v_2 \in
\VS_2\) for which
\[
  v_1 \not \prec_{\Ww} v_2 \; \text{and} \;
  v_2 \not \prec_{\Ww} v_1 \;  \quad 
  \Longleftrightarrow \quad
  \vec v_2 - \vec v_1 \in \RealNum^{2} \setminus (C \cup (-C))
  =: \Xi,
\label{eq:inacc}\numberthis
\]
where a ``causally forbidden'' region \(\Xi\) appears, which has vanishing measure only if~\(\beta = 0\).

  \subsection{Haag-Ruelle Theorem with Adapted Lorentz Frame} \label{sec:generalW}
  Difficulties as in \eqref{eq:inacc} result from the implicit Lorentz-frame dependence of
  the Haag-Ruelle operators~\(\BB_\T(f)\). Nevertheless the latter 
  were well suited for the case of upright~\(\Ww\), which motivates us to
  adapt the construction from \Cref{thm:hrw} by passing to a suitable reference
  frame.\footnote{Constructions using Lorentz-covariant creation-operator
    approximants (e.g.\ \cite{Hdg13}) 
    face similar problems as in \eqref{eq:inacc} when applied in a wedge-local setting.}

\begin{Def}[adapted Haag-Ruelle operators]
  For a general (possibly non-upright) wedge~\(\Ww\), \(A \in \Alg(\Ww)\), \(B = A(\chi)\) as before and
  regular positive-energy Klein-Gordon solutions~\(f\), we set for \(\T \in
  \RealNum\)
  \begin{align*} 
    \BB_\T^\Lambda(f) : = \int \DInt[s] x \; f(\Lambda(\T, \vec x))
    \alpha_{(\Lambda(\T, \vec x))} (\Bb),
    \label{eq:defHRgen} \numberthis
  \end{align*} 
  where \( \Lambda \in \LorentzGroup^*(\Ww) := \{ \Lambda \in
  \LorentzGroup_+^\uparrow: \Lambda \WwR = \Ww_\C \} \)
  or more generally  \(\Lambda \in \LorentzGroup^\uparrow_+\).
\end{Def}

In fact, such \(\BB_\T^\Lambda(f)\) appear naturally in the discussion of
Lorentz covariance in standard Haag-Ruelle theory. Here we just introduce
them in an ad-hoc manner, even if the wedge-local net may not be Lorentz covariant.
In the following we will see that they can equally well serve as
creation-operator approximants, which will turn out suitable for our cause.
We should emphasize that no Lorentz transformation is applied to~\(\Bb\) --- only
the hyperplane used for smearing the translates~\(\alpha_x(\Bb)\) is modified.
Fortunately it is not necessary to repeat our arguments from
\Cref{sec:proofhrw}.
We will instead infer the existence of the limits
\[
  \Psi^+_\Lambda :=
  \lim_{\T\to\infty} \BB_{1\T}^\Lambda(f_1) \BB_{2\T}^\Lambda(f_2)
  \ldots \BB_{n\T}^\Lambda(f_n) \Omega
  \label{eq:defOutL}\numberthis
\]
and their Fock structure for suitably ordered wave packets from a redefinition
of the wedge-local net and the results of \Cref{sec:proofhrw}. The basic
observation is that the modification of passing from \(f\) to
\(f^\Lambda(x) := f(\Lambda x)\) and from translation by \(\alpha_x\) to
modified translation automorphisms \(\alpha^\Lambda_x := \alpha_{\Lambda x}\)
entering in \eqref{eq:defHRgen} are both compatible with the underlying
structures in a sense to be made precise now.
  
  \begin{Lem} Let \(\Lambda \in \LorentzGroup^\uparrow_+\). \label{prop:elemcov}
    \begin{enumerate}[(i)]
      \item \(f^\Lambda(x) := f(\Lambda x)\) defines a regular positive-energy Klein-Gordon solution
        iff \(f\) is a regular positive-energy Klein-Gordon solution.
        \label{it:kgCov}
      \item
        Setting \(\alpha^\Lambda_x := \alpha_{\Lambda 
  x}\) and \(\Alg^\Lambda(\Ww) := \Alg(\Lambda \Ww)\), 
        \((\Alg^\Lambda, \alpha^\Lambda, \Omega)\) is a wedge-local
        quantum field theory satisfying \eqref{eq:HK1}--\eqref{eq:HK6}, and
        possibly \eqref{eq:HK2s}, \eqref{eq:HK3s}, iff  the corresponding assumptions
        hold for \((\Alg, \alpha, \Omega)\).  \label{it:netCov}\\
        If \eqref{eq:HK3s} holds, we set further
      \(\alpha_\lambda^\Lambda(A) := \alpha_{(\Lambda x, \Lambda \Lambda_1\Lambda^{-1})}(A)\)
      for
    \(\lambda = (x, \Lambda_1) \in \PoincareGroup^\uparrow_+\).
    \end{enumerate}
    \proof Lorentz invariance of the Klein-Gordon equation~\itref{it:kgCov} is
    standard, so let us only comment that the restriction to orthochronous
    Lorentz transformation is essential for
    preserving the positive-energy property, and that the regularity property can
    be concluded via the representation \eqref{eq:defKG} and standard 
    (non-)stationary phase estimates.

    Statement \itref{it:netCov} follows from elementary computations which we illustrate 
    for the example of \eqref{eq:HK3s}. Letting
    \(\lambda = (x, \Lambda_1) \in \PoincareGroup^\uparrow_+\) we 
    obtain     \[
      \alpha_\lambda^\Lambda \Alg^\Lambda(\Ww) = 
      \alpha_{(\Lambda x, \Lambda \Lambda_1\Lambda^{-1})}\Alg(\Lambda\Ww)
      = 
      \Alg(\Lambda\Lambda_1\Lambda^{-1}\Lambda\Ww+\Lambda x) = 
      \Alg^\Lambda(\Lambda_1\Ww + x) 
    \]
    where we used that \eqref{eq:HK3s} holds for
    the original net \(\Alg\). \qed
  \end{Lem}
It should be noted that \Cref{prop:elemcov} \itref{it:netCov} applies also to
wedge-local nets which are not Poincaré covariant~\eqref{eq:HK3s}.
In particular the basic definitions \(\alpha^\Lambda_x := \alpha_{\Lambda x}\)
and \(\Alg^\Lambda(\Ww) := \Alg(\Lambda \Ww)\), do not make use of
Lorentz-transformation isomorphisms, they are only a passive redefinition on the
level of the wedge-local net.

To establish the Haag-Ruelle theorem for the adapted scattering state
approximants in \eqref{eq:defOutL} we rewrite the adapted Haag-Ruelle
operators in terms of the boosted net of \Cref{prop:elemcov} as
\begin{align*}
  \BB_\T^\Lambda(f) &= \int \DInt[s] x \; f^\Lambda(\T, \vec x)
  \alpha^\Lambda_{(\T, \vec x)} (\Bb), \text{ and similarly}
  \label{eq:defHRgen2} \numberthis
  \\
  \Bb &= \Aa(\chi) = \int \DInt[d] x' \, \chi(x') \alpha_{x'}(\Aa)
    \\&= \int \DInt[d] x \,\chi(\Lambda x) \alpha_{\Lambda x} (\Aa)
       = \int \DInt[d] x \,\chi^\Lambda( x) \alpha^\Lambda_{x} (\Aa),
\end{align*}
where we used Lorentz invariance of \(\DInt[d] x\). Due to
\(\chi^\Lambda(x):=\chi(\Lambda x) \in \SchwartzSpace(\RealNum^\dd)\) we know
that \(B\) is almost wedge-local also for the redefined net \(\Alg^\Lambda\).
Therefore~\Cref{thm:hrw} may be applied to the rewritten
operators~\eqref{eq:defHRgen2}.
It remains to rephrase the statement of \Cref{thm:hrw} from the boosted
net~\((\Alg^\Lambda,\alpha^\Lambda,\Omega)\) to return to the terminology of
the original theory~\((\Alg, \alpha, \Omega)\). 

Let \(\Ww\) be any wedge, \(\Psi_j = E_m A_j \Omega = E_m A_j^\perp \Omega\),
\(A_j \in \Alg(\Ww)\), \(A_j^\perp \in \Alg(\Ww^\perp)\), and \(\Lambda \in
\LorentzGroup_+^\uparrow\).
Then \(\Psi_j\) are obviously also swappable with respect to the boosted net 
and in particular \(A_j \in \Alg^\Lambda(\Lambda^{-1}\Ww)\).
For~\(\Lambda\in \LorentzGroup^*(\Ww)\) we get \(A_j \in \Alg^\Lambda(\WwR)\),
where \(\Lambda^{-1} \Ww_\C = \WwR\) is upright. Hence assuming uprightness is
redundant for the adapted Haag-Ruelle construction with \(\Lambda\in \LorentzGroup^*(\Ww)\).
Secondly we see from \eqref{eq:defHRgen2} that applying \Cref{thm:hrw}
to outgoing scattering-state approximants interpreted via the boosted net now requires
the ordering
\[
  \VS_{f_n^\Lambda}\prec_\WwR \VS_{f_{n-1}^\Lambda}\prec_\WwR \ldots
  \prec_\WwR \VS_{f_1^\Lambda},
  \label{eq:ordPrescBL}\numberthis
\] 
with \(\VS_{f_j^\Lambda}\) as in \eqref{eq:defVelo}, denoting the
velocity support of \(f_j^\Lambda(x) := f_j(\Lambda x)\). In terms of the
original net, \eqref{eq:ordPrescBL} is by covariance of the ordering relation
(\Cref{prop:ordCov})  equivalent to
\[
  \Lambda\VS_{f_n^\Lambda}\prec_\Ww \Lambda\VS_{f_{n-1}^\Lambda}\prec_\Ww \ldots
  \prec_\Ww \Lambda\VS_{f_1^\Lambda}.
  \label{eq:ordPrescRL}\numberthis
\] 
This is also consistent with a corresponding localization of the adapted
Haag-Ruelle operators~\eqref{eq:defHRgen2} similarly as in
\Cref{lem:wedgeLocalBT}, but with respect to adapted velocity supports
\[
  \VS_{f_j}^\Lambda := \Lambda \VS_{f_j^\Lambda}.
    \label{eq:modVeloSupp}\numberthis
\]
The result of this discussion will be summarized in \Cref{thm:hrwGen}.

\begin{Thm}
  \label{thm:hrwGen}
  Let \(\Lambda \in \LorentzGroup^\uparrow_+\) and \(\Psi_j = E_m A_j \Omega =
   E_m A_j^\perp \Omega\) with \(A_j \in \Alg(\Ww)\), \(A_j^\perp \in
  \Alg(\Ww^\perp)\) (as in \Cref{thm:hrw}).
  \begin{enumerate}[(i)]
    \item \label{it:hrwgconv}
  For regular positive-energy Klein-Gordon solutions \(f_j\) satisfying
\[
  \VS_{f_n}^\Lambda \prec_\Ww \VS_{f_{n-1}}^\Lambda \prec_\Ww \ldots
  \prec_\Ww \VS_{f_1}^\Lambda,
  \label{eq:ordPrescL}\numberthis
\] 
the scattering state approximants
\(
  \Psi_\T^\Lambda := \BB_{1\T}^\Lambda(f_1) \BB_{2\T}^\Lambda(f_2) \ldots
  \BB_{n\T}^\Lambda(f_n)\Omega
\)
converge in norm for \(\T\to \infty\).

\item \label{it:hrwgfock}
  For \(\Lambda \in \LorentzGroup^*(\Ww)\) 
  scalar products of \(\Psi^+_\Lambda :=
    \lim_{\T\to\infty} \BB_{1\T}^\Lambda(f_1) 
    \ldots \BB_{n\T}^\Lambda(f_n) \Omega
  \),
 \(\Psi'^+_\Lambda :=
    \lim_{\T\to\infty} \BB_{1\T}'^\Lambda(f_1') 
    \ldots \BB_{n'\T}'^\Lambda(f_{n'}') \Omega
  \) constructed w.r.t.\ the same wedge \(\Ww\)
  satisfy
  \[
\left\langle \Psi^+_\Lambda, \Psi'^+_\Lambda\right \rangle
= \delta_{nn'} \prod_{j=1}^n \left\langle \BB_{j\T}^\Lambda(f_j)\Omega,
\BB_{j\T}'^\Lambda(f_j')\Omega \right\rangle.
\label{eq:fockGen}\numberthis
\]
\end{enumerate}
Analogous statements hold for incoming scattering states assuming opposite ordering.
\end{Thm}

\subsection{Lorentz-Frame Independence and Residual Covariance}
For the adapted creation-operator approximants~\(\BB_{j\T}^\Lambda(f_j)\),
convergence of approximants \(\Psi_\T^\Lambda := \BB_{1\T}^\Lambda(f_1)
\ldots \BB_{n\T}^\Lambda(f_n)\Omega\) has now been established for general
wedges, i.e.\ upright or tilted. The new ordering
restrictions~\eqref{eq:ordPrescL} appear optimal in the context of
\Cref{ex:nonuprightw},
and the Fock structure follows without additional assumptions.
However, as in standard Haag-Ruelle theory, the choice
HR-operators~\(\BB_{j\T}^\Lambda(f_j)\) creating a given one-particle
vector~\(\Psi_j = \BB_{j\T}^\Lambda(f_j)\Omega\) is not unique. 
Fock structure (\Cref{thm:hrwGen} \itref{it:hrwgfock}) implies only for
fixed~\(\Lambda\), that resulting scattering states do not depend on this
freedom of choosing~\(\BB_{j\T}^\Lambda(f_j)\).  In the following we will
exclude also any unphysical dependence on~\(\Lambda\in \LorentzGroup^*(\Ww)\),
for which one has to handle the non-trivial dependence of localization of
\(\BB_{j\T}^\Lambda(f_j)\) on \(\Lambda\).  We begin by considering the
\(\Lambda\)-dependence of one-particle vectors, to be followed by discussing the
influence on ordering conditions and finally on scattering states. 

\begin{Lem} \label{prop:L1} 
  Let \(\Lambda \in \LorentzGroup_+^\uparrow\) and \(f\) a regular
  positive-energy Klein-Gordon solution.
  \begin{enumerate}[(i)]
    \item The wave packet of \(f^\Lambda(x) := f(\Lambda x)\) 
      as defined in \eqref{eq:defKG} 
      is given by
      \[ 
        \tilde f^\Lambda(\vec k) =
        \frac{ \omega_m(\vec \Lambda_m(\vec k))}{\omega_m(\vec k)} \tilde f(\vec \Lambda_m(\vec k)), 
\label{eq:Lftilde}\numberthis
      \]
 where \(\tilde f\) is the wave
 packet of~\(f\) and
 \(\vec \Lambda_m(\vec k)\) denotes the spatial part of \(
 \Lambda \cdot (\omega_m(\vec k) , \vec k)\). 
 In particular, \(\support \tilde f^\Lambda = \vec \Lambda_m^{-1}(\support \tilde
 f)\).
    \item
  The
  \(\Lambda\)-dependence of 
one-particle vectors is 
  \begin{align*}
    \BB_{\T}^\Lambda(f) \Omega = 
    \frac{\omega_m(\Pp)} {\omega_m(\vec \Lambda^{-1}_m(\Pp))}
    \tilde f(\Pp)
    E(H_m) \Bb\Omega. \label{eq:LOne} \numberthis
  \end{align*}
 \end{enumerate}
\end{Lem}
These one-particle covariance formulas are well-known from the discussion of
Lorentz-covariance in the local case and we will only briefly sketch the
computations in \Cref{app:proofs}.
They are important for the present discussion, as
\eqref{eq:LOne} suggests a non-trivial dependence of \(\lim_{\T\to\infty} \Psi_\T^\Lambda\)
on the auxiliary boost~\(\Lambda\). 
However the dependence can be absorbed by passing to Klein-Gordon
solutions~\(f^{(\Lambda)}_j\) defined via modified wave
packets~\(\tilde f^{(\Lambda)}_j(\vec p) 
  :=  \frac  {\omega_m(\vec \Lambda^{-1}_m(\vec p))} {\omega_m(\vec p)} \tilde f_j(\vec p) \),
which have identical velocity supports and give via \eqref{eq:LOne} that
\[
  \BB_{j\T}^\Lambda(f_j^{(\Lambda)}) \Omega = 
  \tilde f_j(\Pp)
  E(H_m) \Bb_j\Omega, \;
  \text{for any }\Lambda \in \LorentzGroup_+^\uparrow.  \label{eq:LOneInd} \numberthis
\]

While the above argument coincides with the familiar result from local QFT, the
discussion of scattering-state dependence requires additional care in the
wedge-local case due to additional ordering requirements.
For brevity reasons we shall focus on \(\Lambda\)-dependence only within the
preferred class of reference frames for a given localization wedge~\(\Ww\)
defined by \( \LorentzGroup^*(\Ww) 
:= \{ \Lambda \in \LorentzGroup_+^\uparrow: \Lambda \WwR = \Ww_\C \} \)
as in~\Cref{thm:hrwGen}.\footnote{Preliminary computations suggest
  that \Cref{thm:hrwGen} also extends to 
  all~\(\Lambda \in \LorentzGroup_+^\uparrow\) as long as the
  ordering~\eqref{eq:ordPrescL} holds for \(\Lambda \in \LorentzGroup^*(\Ww) \).
}
\begin{Rem}
  \label{rem:prefw}
  Clearly any \(\Lambda, \Lambda' \in  \LorentzGroup^*(\Ww)\) are related by
  an element \(\bar \Lambda := \Lambda^{-1} \Lambda^{'}\) from the 
  stabilizer~\(\mbox{Stab}_{\LorentzGroup_+^\uparrow} \WwR  := \{ \Lambda \in
  \LorentzGroup_+^\uparrow : \Lambda \WwR = \WwR\} \cong 
   \OO(1,1)_+^\uparrow \times \SO(d-2)\), where the first factor is generated by boosts
  \(\Lambda_\beta\) in \(x^1\)-direction (\(\beta \in \RealNum\)), and the
  second by rotations fixing~\(x^1\). 
  In particular we note for later reference that \(\mbox{Stab}_{\LorentzGroup_+^\uparrow} \WwR\) is
  path connected, and that we may smoothly interpolate 
  between any \(\Lambda , \Lambda' \in \LorentzGroup^*(\Ww)\) via arbitrarily
  often differentiable maps \(\Lambda^\gamma : [0,1] \to \LorentzGroup^*(\Ww)\)
  such that \(\Lambda^0 = \Lambda\), \(\Lambda^1 = \Lambda'\).
\end{Rem}
\begin{Prop}[\(\LorentzGroup^*(\Ww)\)-invariance of velocity ordering]
  \label{prop:ordInd}
  For regular Klein-Gordon solutions \(f_1, f_2\) and any \(\Lambda, \Lambda'
  \in \LorentzGroup^*(\Ww)\) we have
  \[
    \VS_{f_1}^\Lambda \prec_\Ww 
    \VS_{f_2}^\Lambda 
    \Longleftrightarrow
    \VS_{f_1}^{\Lambda'} \prec_\Ww 
    \VS_{f_2}^{\Lambda'}
    \label{eq:ordInv}\numberthis
  \]
 \proof By \Cref{prop:ordCov} we have
  \( \VS_{f_1}^\Lambda \prec_\Ww \VS_{f_2}^\Lambda  \Longleftrightarrow 
  \VS_{f_1^\Lambda} \prec_\WwR \VS_{f_2^\Lambda} \)
  and similarly for \(\Lambda'\), allowing us to reduce
  \eqref{eq:ordInv} to the case \(\Ww = \WwR\) up to boosts acting on \(f_j\).
  Thus~\eqref{eq:ordInv} amounts to a property of the relativistic
  velocity transformation law. Let us assume that
  \(\VS_{f_1^{\Lambda'}} \prec_\WwR \VS_{f_2^{\Lambda'}}\).
  By \Cref{rem:prefw} we may write \(\Lambda ' = \Lambda \bar \Lambda\),
  \(\bar \Lambda = \Lambda_\beta R_1\) with a boost~\(\Lambda_\beta\)  in
  \(x^1\)-direction of rapidity \(\beta \in \RealNum\)  and a spatial
  rotation~\(R_1\) preserving \(x^1\). Hence from~\(f^{\Lambda'}_j = f^{\Lambda
    \Lambda_\beta R_1}_j =
(f^{\Lambda\Lambda_\beta}_j)^{R_1} \), (\(j=1,2\)),
  we obtain for the spatial 
  projection~\(\vVS_{f^{\Lambda'}_j}\) of \(\VS_{f^{\Lambda'}_j}\)
  that 
  \begin{align*}
    \vVS_{f^{\Lambda'}_j}
    &=\left \{ \frac{\vec k}{\omega_m(\vec k)}, \; \vec k \in   R_1^{-1}(\support
\tilde f_j^{\Lambda\Lambda_\beta}) \right \}
     = 
     \left  \{ \frac{R_1^{-1} \vec k}{\omega_m(\vec k)}, \; \vec k \in \support
    \tilde f_j^{\Lambda\Lambda_\beta} \right \}
    = R_1^{-1} \vVS_{f^{\Lambda \Lambda_\beta}_j}.
  \end{align*}
  Here we used \Cref{prop:L1} (i),
  that 
  \(R\) from the rotation subgroup of \(\LorentzGroup_+^\uparrow\) 
  act on~\(H_m\) by \(\vec R_m(\vec k) = R \vec k\), and \(\omega_m(R_1^{-1}
  \vec k) = \omega_m(\vec k)\).
  By covariance (\Cref{prop:ordCov})
    \[
      \VS_{f_1^{\Lambda'}} \prec_\WwR \VS_{f_2^{\Lambda'}}
    \Longleftrightarrow 
    R_1^{-1}\VS_{f_1^{\Lambda \Lambda_\beta}} \prec_\WwR R_1^{-1} \VS_{f_2^{\Lambda\Lambda_\beta} }
    \Longleftrightarrow 
    \VS_{f_1^{\Lambda \Lambda_\beta}} \prec_{\WwR}
    \VS_{f_2^{\Lambda\Lambda_\beta} },
  \]
  where we used that \(R_1\WwR = \WwR\), as \(R_1\) is also a rotation preserving \(x_1\).
  The remaining \(x^1\)-boost gives 
  \begin{align*} 
    \vVS_{(f_j^{\Lambda})^{\Lambda_\beta}}
      &=\left \{ \frac{\vec k}{\omega_m(\vec k)}, \; 
        \vec k \in ({\vec  \Lambda}_\beta)_m^{-1}(\support \tilde f_j^{\Lambda})\right \}
        =\left \{ \frac{(\vec \Lambda_{-\beta})_m(\vec k)}
        {\omega_m((\vec \Lambda_{-\beta})_m(\vec k))}, 
          \; \vec k \in \support \tilde f_j^{\Lambda}\right \}
      \\
      &= \left\{ 
      \frac {((\sinh (-\beta) \omega_m(\vec k) + \cosh(-\beta) k^1), k^2, \ldots, k^s)}
      { \cosh (-\beta) \omega_m(\vec k) + \sinh(-\beta) k^1}, \;
      \vec k \in \support
      \tilde f_j^{\Lambda} \right\},
  \end{align*}
  where we used the group action property~\((\vec \Lambda_\beta)^{-1}_m(\vec k) = 
  (\vec \Lambda_\beta^{-1})_m(\vec k) = (\vec \Lambda_{-\beta})_m(\vec k)\). 
  From this we obtain
  \(\VS_{f_1^{\Lambda \Lambda_\beta}} \prec_{\WwR}
  \VS_{f_2^{\Lambda\Lambda_\beta} } \Longleftrightarrow
  \forall \vec k_2 \in \support \tilde f_2^\Lambda,  \;
  \vec k_1 \in \support \tilde f_1^\Lambda :
  \)
  \[
    \frac{-\sinh (\beta) \omega_m(\vec k_2) + \cosh(\beta) k_2^1}
      { \cosh (\beta) \omega_m(\vec k_2) - \sinh(\beta) k_2^1}
      -
    \frac{- \sinh (\beta) \omega_m(\vec k_1) + \cosh(\beta) k_1^1}
      { \cosh (\beta) \omega_m(\vec k_1) - \sinh(\beta) k_1^1} >0.
  \]
  Passing to the common denominator and using~\(\cosh(\beta)^2 - \sinh(\beta)^2 =
  1\), this is equivalent to \(
  {k_2^1}/{\omega_m(\vec k_2)} - 
      {k_1^1}/{\omega_m(\vec k_1)}> 0
  \). As the equivalence holds for all 
  \( \vec k_2 \in \support \tilde f_2^\Lambda,  \;
  \vec k_1 \in \support \tilde f_1^\Lambda \), we have shown that
  \( \VS_{f_1^{\Lambda}} \prec_{\WwR}
  \VS_{f_2^{\Lambda} } \).\qed
\end{Prop}

  This establishes that all choices \(\Lambda \in \LorentzGroup^*(\Ww)\) are
  equivalent with respect to the ordering restriction. That is  a prerequisite
  for the following commutator estimate, which extends \Cref{cor:commEst}
  and will be required for comparing scattering states defined for distinct
  \(\Lambda\in \LorentzGroup^*(\Ww)\). 

  \begin{Lem}[commutator decay] \label{cor:commDecL}
    Let \(A \in \Alg(\Ww), A^\perp \in \Alg(\Ww^\perp)\), and 
      \(f, f^\perp\)  s.t.\
      \(\VS_{f^\perp}^\Lambda \prec_\Ww \VS_{f}^\Lambda\) for some
    \(\Lambda \in \LorentzGroup^*(\Ww)\). Then
 for any compact continuously differentiable curve
 \(\Lambda^\gamma \in \LorentzGroup^*(\Ww)\), \(\gamma \in [0,1]\),
    and \(\T > 0\),
    we have
    \( \norm{[\partial_\gamma \BB_\T^{\Lambda^\gamma}(f^{(\Lambda^\gamma)}),
        \BB_\T^{\perp\Lambda^\gamma}(f^{\perp(\Lambda^\gamma)})]}
    \leq C_N \T^{-N}\) uniformly in \(\gamma\),
     with \(f^{(\Lambda)}\) is as in
    \eqref{eq:LOneInd}.
   \proof
   \def\LG{{\Lambda^\gamma}}
   \def\BLG{{(\Lambda^\gamma)}}
   As before,
   \(\BB_{\T}^{\perp\LG}(f^{\perp\BLG})\)
   may be understood as a creation operator with \(\Lambda = \Id\) with respect
   to the family of boosted theories
   \((\Alg^{\Lambda^\gamma},\alpha^{\Lambda^\gamma},\Omega)\)
   from \Cref{prop:elemcov}.
   Therefore \Cref{lem:wedgeLocalBT} applies and yields wedge-local
   approximants~\((\BB^{\perp \gamma}_\T)^{(\delta)} := (\BB^{\perp
   \LG}_\T(f^{\perp\BLG}))^{(\delta)}\)
       (\(\delta>0\))
       such that for any \(N \in \NaturalNum\),
   \(\normm{ (\BB^{\perp \gamma }_\T)^{(\delta)} -
   \BB^{\perp \LG}_\T(f^{\perp\BLG}) } < C_N^\gamma/(1+\T^{N})\)
   with \( (\BB^{\perp \gamma}_\T)^{(\delta)} \in \Alg(\Ww^\perp +
   \T\VS_{f^\perp}^{\Lambda^\gamma} + \DoubleCone_{\delta\abs{\T}})\).
   Here we already used that \(\VS_{f^{\perp\BLG}}^{\Lambda^\gamma} =
     \VS_{f^\perp}^{\Lambda^\gamma}\) 
     holds for all \(\gamma\),  as the supports of the packets of \(f^\perp\) and
     \(f^{\perp \BLG}\) coincide
    by definition.
    Additionally due to compactness and continuous \(\gamma\)-dependence of
    \(f^{\perp\BLG}\) 
    we can in fact chose \(C_N= C_N^\gamma \) uniformly in \(\gamma \in [0,1]\). 
   For the second operator we similarly note that 
     \begin{align*}
       \partial_\gamma \BB^{\LG}_\T(f^{\BLG}) 
       &=
       \BB_\T^\LG(\partial_\gamma f^{(\Lambda^\gamma)})
       +
       \int\DInt[\s]x \, \left( \big.
         (\partial_\mu \Bb)(\Lambda^\gamma(\T, \vec x)) f^{(\Lambda^\gamma)}(\Lambda^\gamma(\T,
         \vec x)) \right.
       \\&\qquad \qquad \qquad \qquad \qquad 
         + \left.
       \Bb(\Lambda^\gamma(\T, \vec x)) (\partial_\mu f^{(\Lambda^\gamma)})(\Lambda^\gamma(\T,
     \vec x))\right) w^\mu_{(\T, \vec x, \gamma)},
     \end{align*}
    with implied summation over~\(\mu\), and where
    \(w^\mu_{(\T, \vec x, \gamma)} := (\partial_\gamma \Lambda^\gamma(\T , \vec x))^\mu\) 
    satisfies by continuous differentiability and compactness of \(\gamma \to
    \Lambda^\gamma\) the bound \(\abss{w^\mu_{(\T, \vec x, \gamma)}} \leq C(\abs \T + \abs{\vec x})\). 
    Therefore \Cref{lem:wedgeLocalBT} applies and yields 
    \( (\partial_\gamma\BB^{\gamma}_\T)^{(\delta)} \in \Alg(\Ww +
    \T\VS_{f}^{\Lambda^\gamma} + \DoubleCone_{\delta \abs{\T}})\)
    such that
   \(\norm{ (\partial_\gamma\BB^{\gamma}_\T)^{(\delta)} -
   \partial_\gamma \BB^{ \gamma}_\T(f^{\BLG}) } < C'_N/(1+\abs{\T}^{N})\),
     where~\(C_N'\) are uniform in \(\gamma\).
     Finally, the commutator estimate follows from the proof of
     \Cref{cor:commEst} in \Cref{app:proofs}.
    \qed
   \end{Lem}

 \begin{Thm}[\(\Lambda\)-independence of scattering states]
   \label{prop:LIndep}
   Assume that for some \(\Lambda_0 \in \LorentzGroup^*(\Ww)\)
   \(
    \VS_{f_n}^{\Lambda_0} \prec_\Ww \VS_{f_{n-1}}^{\Lambda_0} \prec_\Ww \ldots
    \prec_\Ww \VS_{f_1}^{\Lambda_0}.
   \)
    Then for any \(\Lambda' \in \LorentzGroup^*(\Ww)\) the scattering states 
    \[
      \Psi^+_{(\Lambda')} := \lim_{\T \to \infty}
      \BB_{1\T}^{\Lambda'}(f_1^{(\Lambda')})
      \ldots
      \BB_{1\T}^{\Lambda'}(f_n^{(\Lambda')}) \Omega
      \label{eq:scattL}\numberthis
    \]
    are well-defined and the limit is independent of  \( \Lambda'\).
   \proof 
   Convergence follows from \Cref{prop:ordInd} and \Cref{thm:hrwGen}.
   Using the above preparations we can establish \(\Lambda\)-independence by
   generalizing the arguments familiar from the local case. Due to \Cref{rem:prefw} we can interpolate
   between the two reference frames specified by~\(\Lambda^0=\Lambda_0\) and \(\Lambda^1=\Lambda'\)
 with a differentiable curve \(\Lambda^\gamma \in
   \LorentzGroup^*(\Ww)\), \(\gamma \in [0,1]\).
    Now we estimate for \(\T > 0\) inductively  with respect to the particle
    number~\(n\) that 
   \begin{align*}
    \norm{\Psi^{(\Lambda)}_\T - \Psi^{(\Lambda')}_\T}
    \leq \int
    _0^1 \DInt \gamma \; \norm{\partial_\gamma
    \Psi^{(\Lambda^\gamma)}_\T } \leq C_N \T^{-N}.
   \end{align*}
   For~\(n=1\) this follows from \eqref{eq:LOneInd}  with \(C_N
   = 0\).
   The induction step is established 
   by
   expanding~\(\normm{\partial_\gamma \Psi^{(\Lambda^\gamma)}_\T} \leq
     \norm{(\partial_\gamma (\BB_{1\T}^\gamma \ldots \BB_{n-1\,\T}^\gamma))
     \BB_{n\T}^\gamma \Omega} +
     \norm{\BB_{1\T}^\gamma \ldots \BB_{n-1\,\T}^\gamma \partial_\gamma \BB_{n\T}^\gamma \Omega} 
     \) where we abbreviated \(\BB_{j\T}^\gamma :=
   \BB_{j\T}^{\Lambda^\gamma}(f_j^{(\Lambda^{(\gamma)})})\).
   Here the second term vanishes due to \eqref{eq:LOneInd} and the first term
   may be estimated by swapping 
   \begin{align*}
     \norm{(\partial_\gamma (\BB_{1\T}^\gamma \ldots \BB_{n-1\,\T}^\gamma)) \BB_{n\T}^\gamma \Omega}
       &=
     \norm{(\partial_\gamma (\BB_{1\T}^\gamma \ldots \BB_{n-1\,\T}^\gamma))
     \BB_{n\T}^{\perp\gamma} \Omega}
     \\&
     \leq\norm{\BB_{n\T}^{\perp\gamma} } \norm{(\partial_\gamma
       (\BB_{1\T}^\gamma \ldots \BB_{n-1\,\T}^\gamma))
       \Omega}
       \\ & \qquad\quad+  \norm{[ 
         \partial_\gamma (\BB_{1\T}^\gamma \ldots \BB_{n-1\,\T}^\gamma),
\BB_{n\T}^{\perp\gamma}]}
     \end{align*}
     where both terms are rapidly decreasing in \(\T\). For the first term 
     this is obtained from the induction assumption and 
     \(\normm{\BB^{\perp \gamma}_{j\T}} \leq C (1+\abs{\T}^{s/2})\) (uniformly in
     \(\gamma \in  [0,1]\)). The second term is estimated
     by expansion of the commutator similarly as in 
     \eqref{eq:expComm}, using \Cref{cor:commDecL} and polynomial bounds
     including \(\normm{\partial_\gamma \BB^{\perp \gamma}_{j\T}} \leq C
     (1+\abs{\T}^{s/2+1})\). \qed
  \end{Thm}

\section{Wave Operators, S-Matrix, and Wedge Transitions} \label{sec:wave}

We have now sufficient understanding of the construction from \Cref{sec:generalW}
to begin with a general and model-independent analysis of the
multi-particle scattering data in wedge-local models.
In particular we propose a formalism for wave operators and S-matrices,
which emphasizes the potential physical peculiarities of multi-particle
scattering in the wedge-local setting.
These considerations will provide the foundation for the study of the
multi-particle structure of the Grosse-Lechner model and related
wedge-local theories in subsequent work.

Guided by conventional Haag-Ruelle theory we additionally need to address restrictions of
our construction regarding swapping and ordering conditions.
Regarding the former it will be convenient to introduce in addition to
the one-particle space~\(\HilbertSpace_1 := E_m \HilbertSpace\) 
the (non-closed) subspaces
\begin{align*}
  \HilbertSpace_1^\Ww &:= \{ 
\Psi_1 \in \HilbertSpace_1, \; \Psi_1 \; \text{swappable w.r.t.}\;
\Ww+x\;\text{for some} \; x \in \RealNum^d
  \},
  \\
  \HilbertSpace_{1c}^\Ww &:= \{ \tilde f(\Pp) \Psi_1,  \;
  \Psi_1 \in   \HilbertSpace_1^\Ww, \;
\tilde f \in \ContinuousFuncs^\infty_c(\RealNum^\s)
\}.
\label{eq:defh1c}\numberthis
\end{align*}
It is clear that \(\HilbertSpace_1^\Ww = \HilbertSpace_1^{\Ww+y} =
  U(y)\HilbertSpace_1^{\Ww} =
\HilbertSpace_1^{\Ww' + y'} \) for any \(y, y' \in \RealNum^d\) by symmetry of
the definition, and 
if covariance~\eqref{eq:HK3s} applies \(U(\Lambda)
\HilbertSpace_1^\Ww = \HilbertSpace_1^{\Lambda\Ww}\).
Lastly \Cref{lem:bu} shows that wedge-duality~\eqref{eq:HK2s}
yields~\(\closure{\HilbertSpace_1^\Ww} = E(H_m)\HilbertSpace\) for any
wedge~\(\Ww\).
Further independent of duality \(\HilbertSpace_{1c}^\Ww \subset 
\closure{\HilbertSpace_1^\Ww}\) is dense by spectral calculus, but one should
not expect \(\HilbertSpace_{1c}^\Ww\) to be a subspace of
\({\HilbertSpace_1^\Ww}\), cf. \cite{BBS01}~Lemma~3.4.
It is clear by definition that for any one particle vector~\(\Psi_k \in
\HilbertSpace_{1c}^\Ww\) we can find associated creation operators such that 
\(\Psi_k = \BB_{k\T}^\Lambda(f_k) \Omega =\BB_{k\T}^{\perp\Lambda}(f_k)
\Omega\), so that we can proceed to the corresponding ordered scattering
states.
The basic conceptual issue to be addressed in the passage from the Haag-Ruelle construction
to the wave operators and the S-matrix concerns the potential implicit
dependence of scattering states on the choice of creation-operator
approximants~\(\BB_{k\T}^\Lambda(f_k)\).
\begin{Lem}
  \label{lem:exch}
  Let \(A_k, A'_k \in \Alg(\Ww)\) together with KG-solutions~\(f_k, f'_k\) 
  and auxiliary functions~\(\chi\), \(\chi'\in\SchwartzSpace(\RealNum^\dd)\)
  (cf.\ \Cref{lem:HRW})
  such that
  \( \BB_{k\T}^\Lambda(f_k) \Omega = \BB_{k\T}'^{\Lambda'}(f_k') \Omega \)
  with
  \(\VS_{n} \prec_\Ww \VS_{n-1} \prec_\Ww \ldots \prec_\Ww \VS_{1}\)
  where \(\VS_k := \VS^\Lambda_{f_k}\) and analogously for 
  \(\VS_k' := \VS^{\Lambda'}_{f_k'}\),
  \(\Lambda,\Lambda' \in \LorentzGroup^*(\Ww)\).
  Then the outgoing limits \(\Psi^+\), \(\Psi'^+\) of
  \(\Psi_\T := \BB_{1\T}^\Lambda(f_1) \ldots \BB_{n\T}^\Lambda(f_n)\Omega\) and
  \(\Psi_\T' := \BB'^{\Lambda'}_{1\T}(f_1')\ldots
  \BB'^{\Lambda'}_{n\T}(f_n')\Omega\) coincide.  
  The same holds for incoming limits with ordering assumptions replaced by
  \(\VS_{1} \prec_\Ww \VS_{2} \prec_\Ww \ldots \prec_\Ww \VS_{n}\).
 \proof For \(\Lambda = \Lambda'\) we find directly
  \(\norm{\Psi^+-\Psi'^+}^2 
   = \norm{\Psi^+}^2 - 2\Re\langle \Psi^+,\Psi'^+\rangle
   + \norm{\Psi'^+}^2 \). This vanishes, as due to Fock structure
   (\Cref{thm:hrwGen} (ii)) and coinciding one-particle vectors we obtain
  \(\langle \Psi^+,\Psi'^+\rangle = \norm{\Psi^+}^2  = \norm{\Psi'^+}^2\).
  The case of general~\(\Lambda,\Lambda' \in \LorentzGroup^*(\Ww)\) follows from
  the above via \Cref{prop:LIndep}. \qed
\end{Lem}

Further one can make sense of velocity supports and the corresponding 
ordering assumptions without reference to Klein-Gordon solutions. 
For a single-particle state~\(\Psi_1 \in \HilbertSpace_1\) 
the classical propagation region and the corresponding 
\(\Lambda\)-velocity support (\(\Lambda \in \LorentzGroup_+^\uparrow\))
are given in terms of the energy-momentum spectral measure
\( E_{(H, \Pp)}(\Delta)\) (\(\Delta \subset \RealNum^{\s+1}\) Borel) 
by
\[
  \Upsilon_{\Psi_1}  := \{ t\cdot(\omega, \vec k), \; t \in \RealNum, \; (\omega,\vec k) \in \support
  (E_{(H,\Pp)} \Psi_1) \},
\]
\[
  \VS^\Lambda_{\Psi_1} := \Upsilon_{\Psi_1} \cap \Lambda T_1, \qquad
  T_1 := \{(1,\vec x), \; \vec x \in \RealNum^{\s}\}.
  \label{eq:defvelo1}\numberthis
\]
The precursor ordering is
lifted to a relation on one-particle vectors~\(\Psi_1, \Psi_2 \in \HilbertSpace_{1}\)
by setting for \(\Lambda \in \LorentzGroup^*(\Ww)\)
\[\Psi_2 \prec_\Ww \Psi_1 :\Longleftrightarrow
  \VS_{\Psi_2}^\Lambda \prec_\Ww \VS_{\Psi_1}^\Lambda
  \label{eq:defprec1}\numberthis
\]
which is well-defined as a consequence of \Cref{prop:ordInd}.

The multi-particle configurations accessible via our wedge-local Haag-Ruelle
construction can be conveniently expressed by the following notion of ordered
Fock spaces replacing the conventional definition based on bosonic or fermionic
statistics.

\begin{Def}
  The ordered tensor products over one-particle Hilbert
  space~\(\HilbertSpace_1\) with respect to a partial order \(\prec\) on
  \(\HilbertSpace_1\) are defined as closure
  \( {\otimes_\prec^n} \HilbertSpace_1
    := \closure{ {\hat \otimes_\prec^n} \HilbertSpace_1} \)
  of the finite linear spans
  \begin{align*}
    {\hat \otimes_\prec^n} \HilbertSpace_1
    &:= {\Span\{ \Psi_1 \otimes \ldots \otimes \Psi_n, \; \Psi_k \in \HilbertSpace_1,
      \Psi_1 \prec  \Psi_2 \prec \ldots \prec \Psi_n\}}.
    \label{eq:nordtp}\numberthis
  \end{align*}
  Using the conventions~\( {\hat \otimes_\prec^0}\HilbertSpace_1 := \ComplexNum \Omega\), 
  \({\hat \otimes_\prec^1}\HilbertSpace_1 := \HilbertSpace_1\) we obtain
  corresponding ordered Fock spaces 
  \(\FS^\prec(\HilbertSpace_1) := \bigoplus_{n=0}^\infty {\otimes_\prec^n}\HilbertSpace_1\).
  The subspace of finite linear combinations of ordered tensor product vectors
  with \(\Psi_k \in \HilbertSpace_1' \subset \HilbertSpace_1\) will be denoted by
  \(\FS^\prec_0(\HilbertSpace_1') := \hat\bigoplus_{n=0}^\infty{\hat\otimes_\prec^n} \HilbertSpace_1'\).
\end{Def}

To proceed to the scattering data note that
\(\FS_0^{\prec_\Ww}(\HilbertSpace_{1c}^\Ww) \subset \FS^{\prec_\Ww}(\HilbertSpace_{1})\)
is dense and the wave operators are defined by linear extension of the
isometries obtained from the wedge-local Haag-Ruelle construction of
\Cref{thm:hrwGen}. Just as for ordinary bosonic- and fermionic statistics,
unsymmetrized
Fock space~\(\FSu(\HilbertSpace_1) := \bigoplus_{n=0}^\infty \HilbertSpace_1^{\otimes n}\)
provides a common enveloping space into which ordered tensor products and
Fock spaces embed naturally. The possible dependence of scattering states on a
given wedge of reference \(\Ww\), noted by Grosse and Lechner \cite{GL07}, extends
also to multi-particle scattering states and is most consequently expressed
on the level of wave operators.
\begin{Def}[wave operators] \label{def:wave}
  For any given centered wedge \(\Ww\) 
  we set
  \begin{align*}
    \WO_\Ww^+ &: 
    \left\{
      \begin{aligned}
      \FS^{\succ_\Ww}_0(\HilbertSpace_{1c}^\Ww)  &\longrightarrow \HilbertSpace,\\
        \Psi_1 \otimes \ldots \otimes \Psi_n  &\longmapsto 
     \lim_{\T \to\infty} \BB_{1\T}^\Lambda(f_1) \ldots \BB_{n\T}^\Lambda(f_n) \Omega,
    \end{aligned}
    \right.
    \label{eq:defwout}\numberthis
    \\
    \WO_\Ww^- &: 
    \left\{
      \begin{aligned}
        \FS^{\prec_\Ww}_0(\HilbertSpace_{1c}^\Ww)  &\longrightarrow \HilbertSpace,\\
        \Psi_1 \otimes \ldots \otimes \Psi_n  &\longmapsto 
     \lim_{\T \to-\infty} \BB_{1\T}^\Lambda(f_1) \ldots \BB_{n\T}^\Lambda(f_n) \Omega,
    \end{aligned}
    \right.
    \label{eq:defwin}\numberthis
  \end{align*}
  where for \(\Lambda\in\LorentzGroup^*(\Ww)\) suitable
  \(\BB_{k\T}^\Lambda(f_k) \Omega = \Psi_k\) with \(\Bb_k\) swappable and almost
  wedge-local w.r.t.\ the given wedge \(\Ww\) can be obtained for \(\Psi_k \in
  \HilbertSpace_{1c}^\Ww\) via \eqref{eq:LOne}.
\end{Def}

\begin{Prop} \label{prop:woext}
  Assuming wedge-duality~\eqref{eq:HK2s},
  the wave operators \eqref{eq:defwout}, \eqref{eq:defwin} are well-defined and extend to bounded
  linear isometries
\(
\WO_\Ww^+ : \FS^{\succ_\Ww}(\HilbertSpace_1) \longrightarrow \HilbertSpace\),
and
\( \WO_\Ww^- : \FS^{\prec_\Ww}(\HilbertSpace_1) \longrightarrow \HilbertSpace
\). 
\proof Well-definedness of \(\WO_\Ww^+\)  on \(\FS_0^{\succ_\Ww}(\HilbertSpace_{1c}^\Ww) \) 
follows by noting that the computation from the proof of  \Cref{lem:exch}
extends to linear combinations of \(\Psi^+\).
As the Fock structure also implies
isometry of \(\WO^{+}_\Ww\) the wave operators further extend to
the closures \(\FS^{\succ_\Ww}(\HilbertSpace_{1})  =
\closure{\FS_0^{\succ_\Ww}(\HilbertSpace_{1c}^\Ww)}\) by continuity 
and using that \(\closure{\HilbertSpace_{1c}^\Ww} = \HilbertSpace_1\)
(\Cref{lem:bu}). The construction of \(\WO_\Ww^-\) is analogous on the
oppositely ordered spaces. \qed
\end{Prop}

Due to translation covariance it is sufficient to consider~\(\WO^\pm_\Ww\) for
centered wedges \(\Ww = \Lambda \WwR\). In other words we will now see that the
wave operators in fact depend on the wedge~\(\Ww\) only modulo translations.
Given \eqref{eq:HK3s}
this symmetry consideration in fact extends to the full Poincaré group,
whose action~\(U_0(\lambda)\) on \(\FSu(\HilbertSpace_1)\)
is defined by
  \begin{align*} 
     U_0(\lambda) \left(\Psi_1 \otimes
    \Psi_2 \otimes \ldots \otimes \Psi_n\right) 
    &:= ( U(\lambda) \Psi_1) \otimes (U(\lambda) \Psi_2) \otimes \ldots \otimes
(U(\lambda) \Psi_n).
\label{eq:defu0}\numberthis
\end{align*} 
While \(U_0(x)\) preserves velocity-ordered Fock spaces, boosts act in general
non-trivially. Explicitly, \Cref{prop:ordCov} shows that \(U_0(\Lambda)
\FS^{\prec_\Ww}(\HilbertSpace_1) = \FS^{\prec_{\Lambda\Ww}}(\HilbertSpace_1)\),
\(U_0(\Lambda)\FS^{\succ_\Ww}(\HilbertSpace_1) = \FS^{\succ_{\Lambda\Ww}}(\HilbertSpace_1)\), and analogously for
the subspaces \(\FS^{\prec_\Ww}_0(\HilbertSpace_{1c}^\Ww)\).

\begin{Thm} \label{prop:cov}
  For \(\lambda = (a,\Lambda) \in \PoincareGroup_+^\uparrow\)
  we have
  \( 
    \WO_{\Ww+a}^\pm =  
    \WO_{\Ww}^\pm\) and \(
    U(\lambda)\WO_\Ww^\pm 
    = \WO_{\Lambda \Ww}^\pm U_0(\lambda).
\) 
\proof
The first statement follows trivially from translation symmetry of 
\Cref{def:wave}.
For the second statement let us consider only the outgoing case, and note that it is
sufficient to establish the identities for special \(\Psi^+\) of ordered tensor product
form
\begin{align*}
  \Psi^+ &= 
\WO^+_\Ww (\BB_{1\T}^{\Lambda'}(f_1)\Omega \otimes  \ldots \otimes
\BB_{n\T}^{\Lambda'}(f_n) \Omega)
\\&=
\lim_{\T \to\infty}  \BB_{1\T}^{\Lambda'}(f_1) \ldots \BB_{n\T}^{\Lambda'}(f_n) \Omega.
\end{align*}
with auxiliary boost~\(\Lambda' \in \LorentzGroup^*(\Ww)\) 
and velocity supports ordered correspondingly, that is by \(\VS_{f_1}^{\Lambda'} \succ_\Ww
\VS_{f_2}^{\Lambda'} \succ_\Ww \ldots \succ_\Ww \VS_{f_n}^{\Lambda'}\).
From continuity of \(U(\lambda)\), we obtain
\begin{align*}
  U(\lambda)\Psi^+
&= \lim_{\T \to\infty} U(\lambda) \BB_{1\T}^{\Lambda'}(f_1)
  U(\lambda)^* U(\lambda) \ldots
    U(\lambda)^* U(\lambda) \BB_{n\T}^{\Lambda'}(f_n) \Omega.
  \label{eq:covstart}\numberthis
\end{align*} 
 Using \(U(\lambda)U(x) = U(\Lambda x) U(\lambda)\), 
 the adjoint action of \(U(\lambda)\) yields due to 
\begin{align*}
  U(\lambda) \BB_{j\T}^{\Lambda'}(f_j) U(\lambda)^*
  &= \int \DInt[s]x f_j(\Lambda'(\T, \vec x)) \, U(\lambda) \alpha_{\Lambda'(\T,
\vec x)}(\Bb_j) U(\lambda)^*
  \\
  &= \int \DInt[s]x f_j'(\Lambda\Lambda'(\T, \vec x)) \,
  \alpha_{\Lambda\Lambda'(\T, \vec x)}( \Bb_j') =
  \Bb_{j\T}'^{\Lambda\Lambda'}(f_j')
  \label{eq:covBT}\numberthis
\end{align*} 
again a Haag-Ruelle operator 
with  \(\Bb_j' := U(\Lambda,a) \Bb_j U(\Lambda,a)^*\)
from the class of almost-wedge local operators considered in \Cref{lem:HRW}
(with respect to the transformed wedge \(\Lambda \Ww\)) and \(f_j'(x) := f_j(\Lambda^{-1} x)\).
Starting from \eqref{eq:covstart} covariance of \(\WO^+_\Ww\) is now obtained via
\begin{align*} 
  U(\lambda)\Psi^+
 &= 
  \lim_{\T\to\infty}  
  \BB_{1\T}'^{\Lambda\Lambda'}(f_1') \BB_{2\T}'^{\Lambda\Lambda'}(f_2')
  \ldots \BB_{n\T}'^{\Lambda\Lambda'} (f_n') \Omega  
  \\
   &= 
  \WO_{\Lambda'\Ww}^+ ( (\BB_{1\T}'^{\Lambda\Lambda'}(f_1')\Omega)
    \otimes \ldots \otimes ( \BB_{n\T}'^{\Lambda\Lambda'} (f_n') \Omega) )
  \\
  &= \WO_{\Lambda\Ww}^+((U(\Lambda) \BB_{1\T}^{\Lambda'}(f_1) \Omega) \otimes
\ldots \otimes (U(\Lambda)\BB_{n\T}^{\Lambda'}(f_n) \Omega))
  \\&
  = \WO_{\Lambda\Ww}^+ U_0(\Lambda) (\BB_{1\T}^{\Lambda'}(f_1) \Omega \otimes
  \ldots \otimes \BB_{n\T}^{\Lambda'}(f_n) \Omega).
\end{align*}
Here we first used~\eqref{eq:covBT}, well-definedness of the wave-operators
(\Cref{prop:woext}), then again~\eqref{eq:covBT}, and lastly~\eqref{eq:defu0}.
Finally we extend by linearity and continuity to all of
\(\FS^{\succ_\Ww}(\HilbertSpace_1)\), whereby we obtain the covariance identity.
\qed
\end{Thm} 

For local theories \(\WO_\Ww^\pm\) are equivalent 
to the conventional Haag-Ruelle wave operators as a consequence of
\Cref{lem:exch}. Therefore in local theories they must be \(\Ww\)-independent
and Lorentz-covariant up to suitable identification of ordered Fock spaces by
standard arguments.
In the general wedge-local setting on the other hand, a non-trivial dependence
of \(\WO^\pm_\Ww\) on the wedge~\(\Ww\) should be expected, as noticed in
\cite{GL07}. The resulting asymptotic breaking of Lorentz symmetry in higher dimensions
will be strongly model dependent, so that it is beyond the scope of our present
general analysis.
The lesson to be learned is that there must be a residual Lorentz covariance
with respect to the stabilizer of \(\Ww_\C\) in any wedge-local theory.

Finally let us note that also the \(S\)-matrix in wedge-local theories, as
accessible via our construction with suitable ordering restrictions, will
inherit the wedge-dependence of the wave operators.
\begin{Def}[S-matrix and wedge-transition maps]
  The \(S\)-matrices and wedge-transition maps between final and initial states
  are defined as
  \[
    S_{\fin\,\ini}^{\Ww_\fin, \Ww_\ini}:= (\WO_{\Ww_\fin}^+)^* \WO_{\Ww_\ini}^-,
    \quad
    S_{\fin\,\fin}^{\Ww', \Ww}:= (\WO_{\Ww'}^+)^* \WO_{\Ww}^+,
    \quad
    S_{\ini\,\ini}^{\Ww', \Ww}:= (\WO_{\Ww'}^-)^* \WO_{\Ww}^-.
    \label{eq:defsmatrix}\numberthis
  \]
  depending on centered wedges~\( \Ww_\fin, \Ww_\ini, \Ww, \Ww'\) entering in the
Haag-Ruelle construction.
\end{Def}

\begin{Thm} S-matrices and wedge transition maps~\eqref{eq:defsmatrix} are Poincaré-covariant in the
  sense that for \(\lambda = (a,\Lambda) \in \PoincareGroup_+^\uparrow\) we have
  \begin{align*}
    U_0(\lambda)S^{\Ww_\fin, \Ww_\ini}_{fi} U_0(\lambda)^*&= 
    S^{\Lambda \Ww_\fin, \Lambda \Ww_\ini}_{fi},
   \\ 
    U_0(\lambda)S^{\Ww, \Ww'}_{ff} U_0(\lambda)^*&= 
    S^{\Lambda \Ww, \Lambda \Ww'}_{ff},
    &
    U_0(\lambda)S^{\Ww, \Ww'}_{ii} U_0(\lambda)^*&= 
    S^{\Lambda \Ww, \Lambda \Ww'}_{ii}.
  \end{align*}
  If the wave operators are asymptotically complete (i.e.\ have dense range in
    \(\HilbertSpace\)) we have additional transition identities such as
\(
S_{\fin\,\ini}^{\Ww_\fin, \Ww_\ini} =
  S_{\fin\,\fin}^{\Ww_\fin, \Ww_\fin'}  
S_{\fin\,\ini}^{\Ww_\fin', \Ww_\ini'}   
  S_{\ini\,\ini}^{\Ww_\ini', \Ww_\ini}
\).
\proof Covariance identities follow from \Cref{prop:cov}. The wedge-transition 
formula is a consequence of \eqref{eq:defsmatrix} using that asymptotic
completeness and isometry of \(\WO_{\Ww_\fin'}^+\) imply \(\WO_{\Ww_\fin'}^+
(\WO_{\Ww_\fin'}^+)^* = \Id\) and analogously for \(\WO_{\Ww_\ini'}^-\). \qed

\end{Thm}
  
  It is important to highlight that in our construction the 
  localization wedge~\(\Ww\) must agree among all creation operators used to
  define a scattering state. Additionally even if there is a non-trivial overlap 
  between two distinct ordered Fock spaces, for non-vanishing \(\Psi \in
    \FS^{\succ_{\Ww}}(\HilbertSpace_1)
    \cap\FS^{\succ_{\Ww'}}(\HilbertSpace_1)\) one will in general have
    \(\WO^+_\Ww \Psi \not = \WO^+_{\Ww'} \Psi \).
  The analysis of this localization-dependence can be carried much
  further in models with stronger (e.g.\ string-like) localization. In this case
  also scattering states can be constructed for mixed
  string-directions and the dependence on these directions can
  be taken into account on the level of the asymptotic Fock spaces \cite{FGR96}.

  \section{Concluding Remarks}
  
We developed \(N\)-particle scattering theory for general wedge-local quantum
field theories with isolated mass shells. In particular we constructed scattering states 
for arbitrarily many particles, even with reduced localization
information available from wedge-locality.
This implies also that the asymptotic particle structure of wedge-local models with
isolated mass shells must be as rich as for strictly local theories.

This brings us to the problem of asymptotic completeness (AC) which, in spite of
recent progress \cite{Le06t, DT10, DG13}, is largely open both in the local and
wedge-local setting.
Using our construction of \(N\)-particle scattering states, we intend to
establish AC in the wedge-local model of Grosse and Lechner \cite{GL07}. 
This will give the first example of a relativistic theory in 4-dimensional
space-time, which is interacting and asymptotically complete. Furthermore we
expect that the non-trivial \(S\)-matrix of this model will be factorizing,
which is an unusual feature in higher dimensions.
On the other hand also interesting counterexamples to two-particle
asymptotic completeness have recently been constructed in wedge-local setting
\cite{LTU17}, which ought to be instructive also at the
multi-particle level.

It is not known whether the existence of an interpolating wedge-local net
has any consequences on the properties of an \(S\)-matrix beyond the basic
symmetry principles discussed in \Cref{sec:wave}. 
As a first step one may ask whether there is any meaningful generalization of
the LSZ reduction formula for the wedge-local setting, which is the conventional
point of departure for investigating analyticity properties of the S-matrix.
Phrased differently, one may ask in which generality the inverse scattering
problem is solvable within the class of wedge-local models. Here some positive
related results are known for non-local models~\cite{BW84}, or for a certain
class of field theories formulated on Krein spaces~\cite{AG01}. 

Lastly let us point out that a general scattering theory for
massless particles in the wedge-local setting curiously appears to require new
ideas. In particular many of the conventional technical tools may fail without
mass gaps, including energy bounds and clustering estimates which are
indispensable in all previous constructions of scattering states in the local
setting without mass gaps, see e.g.~\cite{Bu77, Dy05, AD15, MD16}.

\appendix
\section{Some Technical Arguments} \label{app:proofs}
For the convenience of the reader we will briefly explain how the standard proof
of commutator estimates for Haag-Ruelle operators also yields the corresponding
results in the wedge-local setting. Due to the covariance arguments from
\Cref{sec:generalW} it is sufficient to consider the case of non-adapted
HR-operators corresponding to \(\BB_{\T}^\Lambda(f)\) with \(\Lambda = \Id\).

\begin{Lem} \label{lem:KGdecay}
  Let \(f\) be a regular Klein-Gordon solution of mass \(m >0\). 
  \begin{enumerate}[(i)]
    \item \( \abs{f(t, \vec x)} \leq C/(1+\abs{t}^{s/2})\) for any \((t, \vec x)
      \in \RealNum^{s+1}\),
    \item \( \abs{f(t, \vec x)} \leq C_{\epsilon, N}/(1+\abs{t}^N + \abs{\vec x}^N)\) for \((t, \vec x) \in
      \RealNum^d \setminus \Upsilon_f^\epsilon\),
    \item \(\norm{f_t}_{\LSpace^1(\RealNum^s)} \leq C(1+\abs{t}^{s/2})\), where
        \(f_t(\vec x) := f(t, \vec x)\),
    \end{enumerate}
    where \(\epsilon > 0\), and \(N \in \NaturalNum\) are arbitrary, \(C>0\),  \(C_{\epsilon,N} >0\) are suitable constants
    depending on \(f\), and \(\Upsilon_f^\epsilon := \RealNum \VS_f^\epsilon\)
    is the cone generated by the \(\epsilon\)-enlarged velocity support
    \(\VS_f^\epsilon := \{ (1, \vec v) \in \RealNum^d, \; \exists \vec (1,\vec v') \in \VS_f, \;
    \abs{\vec v - \vec v'} < \epsilon\}\).
  \proof See \cite[Thm.\ 5.3]{ArQFT99}.\qed
\end{Lem}

\begin{proof}[Proof of \Cref{lem:wedgeLocalBT}] Let \(\delta > 0\) be given and
  \(\Bb_r \in \Alg(\Ww+\DoubleCone_r)\), \(\norm{\Bb - \Bb_r} \leq C_N
  /(1 + r^{N})\)  as in \Cref{lem:HRW}.
  Suitable wedge-local approximants may then be obtained by restricting the
  integration in the definition of \(\BB_\T(f)\) to the asymptotically dominant
  part \(f^\uparrow(x) := f(x) \CharFct_{\Upsilon_f^{\delta/2}}(x)\) 
  (\Cref{lem:KGdecay}) and 
  setting \(r(\T):= \delta \abs{\T}/2\) to obtain
  \begin{align*}
    \BB_\T^{(\delta)} := (\Bb_{r(\T)})_\T(f^\uparrow) = \int \DInt[s]  x \; f^\uparrow(\T, \vec x) 
    \alpha_{(\T, \vec x)} (\Bb_{r(\T)}) \in \Alg(\Ww + \DoubleCone_{\delta \abs{\T}/2}
    + \T \VS_f^{\delta/2}) ,
  \end{align*} 
  where the localization was computed for given \(\T \in
    \RealNum\) by covariance, isotony and  noting that \( \Upsilon_f^{\delta/2} \cap \{ x \in \RealNum^d, \;
    x^0 = \T\} = \T \VS_f^{\delta/2} \subset \T \VS_f +
  \DoubleCone_{\delta \abs{\T}/2}\) and \(\DoubleCone_{\delta \abs{\T}/2}
+ \DoubleCone_{\delta \abs{\T}/2} \subset \DoubleCone_{\delta \abs{\T}}\).
The approximation in norm is established
by \(\normm{\BB_\T(f) - \BB_\T^{(\delta)} }
  = \norm{\Bb_\T(f) -  (\Bb_{r(\T)})_\T(f^\uparrow)}
  \leq \norm{(\Bb - \Bb_{r(\T)})_\T(f)}
  + \norm{(\Bb_{r(\T)})_\T(f-f^\uparrow)}
  \leq \norm{\Bb - \Bb_{r(\T)}} \norm{f_\T}_{\LSpace^1(\RealNum^\s)} +
   \normm{\Bb_{r(\T)}} \normm{f_\T- f^\uparrow_\T}_{\LSpace^1(\RealNum^\s)} 
  \)
  using \(\norm{\Bb_r} \leq \norm{\Bb}+C_1\),
  \( \normm{f_\T- f^\uparrow_\T}_{\LSpace^1(\RealNum^\s)} \leq
  C_N'/(1+\abs{\T}^{N}) \) due to \Cref{lem:KGdecay} and that
  \(\norm{\Bb-\Bb_{r(\T)}} \leq C_N/(1+\delta^N \abs{\T}^N)\) is sufficient to
  compensate the polynomial growth in \Cref{lem:KGdecay} (iii) and obtain
  overall \(\normm{\BB_\T(f) - \BB_\T^{(\delta)} } \leq C_{\delta,N}/(1+\abs{\T}^N)\).
\end{proof}

\begin{proof}[Proof of \Cref{cor:commEst}]
  To estimate \(\norm{[\BB^\perp_\T(f^\perp), \BB_\T(f)]}\), let
  \(\delta > 0\) and \(\BB_\T^{(\delta)}\), \( \BB_\T^{\perp(\delta)}\) corresponding
  approximants as from \Cref{lem:wedgeLocalBT},
  i.e.\ \(\BB_\T^{(\delta)} \in \Alg(\T \VS_f +
  \DoubleCone_{\delta\abs{\T}} + \Ww)\), s.t.\
  \(
    \norm{\BB_{\T}^{(\delta)} - \BB_\T(f)} \leq {
    C_N^\delta}/(1+\abs{\T}^N)\), and let analogously
    \(\BB_\T^{\perp(\delta)} \in \Alg(\T \VS_{f^\perp} +
  \DoubleCone_{\delta\abs{\T}} + \Ww^\perp)\), s.t.
  \(
    \normm{\BB_{\T}^{\perp(\delta)} - \BB_\T^\perp(f^\perp)} \leq {
  C'^\delta_N}/(1+\abs{\T}^N)\).

  Choosing \(\delta > 0\) sufficiently small the localization regions of
  \(\BB_\T^{(\delta)}\) and \(\BB_\T^{\perp(\delta)}\) will
  be space-like separated for any large enough \(\T > 0\):
  By assumption we have \(\VS_{f} - \VS_{f^\perp} \subset \Ww_\C\) with
  \(\VS_{f} - \VS_{f^\perp}\) compact and \(\Ww_\C\) open. Thus there exists \(\epsilon > 0\) such that 
\( \VS_{f} - \VS_{f^\perp} + \DoubleCone_\epsilon \subset \Ww_\C\), where
\(\DoubleCone_\epsilon : = \{ x \in \RealNum^d, \; \abs{x}_c = \abss{x^0} + \abs{\vec x} <
\epsilon \}\) and as \(\Ww_\C\) is a convex cone this implies also
\( \T( \VS_{f} - \VS_{f^\perp} + \DoubleCone_\epsilon) \subset \Ww_\C\) for any
\(\T > 0\).
To obtain space-like separation recall that \(\Ww = \Ww_\C + x_1\), \(\Ww^\perp
  = \Ww_\C' + x_2\), for \(x_1, x_2
  \in \RealNum^d\). Thus we get for \(\delta < \epsilon/3\) and \(\T > 3(\abs{x_1}_c +
  \abs{x_2}_c)/\epsilon\) and any \(\abs{x_1'}_c < \delta\), \(\abs{x_2'}_c
< \delta\) that
\begin{align*}
  \fl
  \T ( \VS_f - \VS_{f^\perp} + \frac{x_1 - x_2}{\T} + x_1' - x_2' ) + \Ww_\C \subset \Ww_\C = (\Ww_\C^\perp)'
  \\&
  \Longrightarrow
  \T \VS_f + x_1 + \T x_1' + \Ww_\C \subset
  (\Ww_\C^\perp + \T \VS_{f^\perp} + x_2 + \T x_2')',
\end{align*}
where we used \(\Ww_\C + \Ww_\C = \Ww_\C\) and that \(\Reg_1 + \Reg_2
\subset \Reg_3' \Longleftrightarrow \Reg_1 \subset (\Reg_3 - \Reg_2)'\) for
arbitrary \(\Reg_k \subset \RealNum^d\).
Due to \(\T \DoubleCone_{\delta} = \DoubleCone_{\delta \T}\) this is equivalent to
\(
  \Ww + \T \VS_f + \DoubleCone_{\delta \T}  \subset (\Ww^\perp + \T
  \VS_{f^\perp} + \DoubleCone_{\delta \T})'
\) for \(\delta < \epsilon/3\) and \(\T > 3(\abs{x_1}_c +
  \abs{x_2}_c)/\epsilon\), as claimed.

  For such \(\T, \delta\) we now obtain from locality that \([
  \BB_\T^{\perp(\delta)}, \BB_\T^{(\delta)}] = 0\), which implies the commutator estimate by
  expanding
  \begin{align*}
  \norm{[\BB^\perp_\T(f^\perp), \BB_\T(f)]}
  & = 
  \norm{[\BB^\perp_\T(f^\perp) - \BB^{\perp(\delta)}_\T + \BB^{\perp(\delta)}_\T  , \BB_\T(f) - \BB^{(\delta)}_\T + \BB^{(\delta)}_\T  ]}
  \\&
  \leq
  \norm{[\BB^\perp_\T(f^\perp) - \BB^{\perp(\delta)}_\T, \BB_\T(f) - \BB^{(\delta)}_\T + \BB^{(\delta)}_\T  ]}
  \\& \qquad\qquad
  + \norm{[ \BB^{\perp(\delta)}_\T  , \BB_\T(f) - \BB^{(\delta)}_\T   ]}
  + \norm{[ \BB^{\perp(\delta)}_\T  ,  \BB^{(\delta)}_\T  ]},
  \end{align*}
  where \(\normm{[\BB^\perp_\T(f^\perp) -
    \BB^{\perp(\delta)}_\T, \BB_\T(f)  ]} \leq 2 \normm{\BB^\perp_\T(f^\perp) -
  \BB^{\perp(\delta)}_\T} \norm{\BB_\T(f)}  \leq 2 C_{N'}^\delta
    C/(1+\abs{\T}^{N'}) \cdot
    (1+\abs{\T})^{s/2} \leq C_N' \T^{-N}
\)
due to \Cref{lem:wedgeLocalBT} and  \Cref{prop:otherProp} \itref{it:nest} and
analogously for the second non-vanishing commutator.
\end{proof}
\begin{proof}[Proof of \Cref{prop:L1}] 
  {\bf Ad (i)}
The wave packet \(\tilde f^\Lambda_\T\) of \(f^\Lambda_\T\)
     can be computed via Fourier inversion theorem by noting that
     \begin{align*}
       f_\T^\Lambda (\vec x) &= 
       \int \frac{\DInt[s]k}{(2\pi)^\s} \; \Ee^{-\Ii (\omega_m(\vec k),
   \vec k)^\mu (\Lambda(\T, \vec x))_\mu} \tilde f(\vec k)
   \\&
        = 
        \int \frac{\DInt[s]k}{(2\pi)^\s\omega_m(\vec k)} \; \Ee^{-\Ii (\Lambda^{-1}(\omega_m(\vec k),
        \vec k))^\mu (\T, \vec x)_\mu} \tilde f(\vec k) \, \omega_m(\vec k)
        \\&=
        \int \frac{\DInt[s]k'}{ (2\pi)^\s\omega_m(\vec k')} \; \Ee^{-\Ii (\omega_m(\vec k'),
          \vec k')^\mu (\T, \vec x)_\mu} \tilde f(\vec \Lambda_m(\vec k')) \,
          \omega_m(\vec \Lambda_m(\vec k')),
     \end{align*}
     where we substituted \(\vec k' := \vec \Lambda_m^{-1}(\vec k)\) after rewriting with
     respect to the standard Lorentz-invariant measure~\(\DInt^{\s}
       k/\omega_m(\vec k)\) (more precisely \(\vec\Lambda_m\)-invariant, see e.g.\ \cite{RS2} Thm.~IX.37)
       and used that \((\Lambda (\omega_m(\vec k), \vec k))^0 = \omega_m(\vec
       \Lambda_m(\vec k))\) due to \((\omega_m(\vec k), \vec k) \in H_m\) and
     Lorentz-invariance of the mass hyperboloid~\(H_m\). 
     
  {\bf Ad (ii)} We obtain
     \begin{align*}
       \BB_{\T}^\Lambda(f) \Omega &= 
      \int \DInt[\s] x \; f^\Lambda(\T, \vec x) \; \Ee^{\Ii (\Lambda^{-1} P)^\mu (\T, \vec
      x)_\mu} \Bb \Omega
      =
      \Ee^{\Ii H_\Lambda \T}
      \int \DInt[\s] x \; f^\Lambda(\T, \vec x) \; \Ee^{- \Ii \Pp_\Lambda \cdot \vec x} \Bb \Omega
      \\&=
      \Ee^{\Ii H_\Lambda \T}
      \int \DInt[\s] x \, \DInt E_{\Pp_\Lambda}(\vec p)\; f^\Lambda_\T( \vec x) \; \Ee^{- \Ii \vec p \cdot \vec x} \Bb \Omega
      = 
      \Ee^{\Ii H_\Lambda \T}
      \tilde f^\Lambda_\T (\Pp_\Lambda) \Bb \Omega.
      \label{eq:LOneComp}\numberthis
     \end{align*}
     Here we first used translation-invariance of \(\Omega\), 
     \(
       P^\mu (\Lambda x)_\mu =
       (\Lambda^{-1} P)^\mu x_\mu
       \), and then we abbreviated \((H_\Lambda, \Pp_\Lambda) := \Lambda^{-1}(H,
     \Pp)\), \(f^\Lambda_\T(\vec x) := f(\Lambda(\T, \vec x))\). Further 
 due to  \eqref{eq:Lftilde},
     \( \tilde f_\T^\Lambda(\vec k) = 
        \frac{ \omega_m(\vec \Lambda_m(\vec k))}{\omega_m(\vec k)} \tilde f(\vec \Lambda_m(\vec k)) 
        \Ee^{- \Ii \omega_m(\vec k) t}
     \), and therefore 
     \( 
      \Ee^{- \Ii \omega_m(\vec P_\Lambda) t}
  \Bb\Omega 
  = \Ee^{- \Ii \omega_m(\vec P_\Lambda) t} E(H_m)\Bb\Omega
  = \Ee^{- \Ii H_\Lambda t} E(H_m)\Bb\Omega
\), so that \(\T\)-dependent terms cancel in \eqref{eq:LOneComp}. 
Finally \eqref{eq:LOne} is obtained by inserting 
\(\vec P_\Lambda = \vec \Lambda_m^{-1}(\vec P) \).
   \end{proof}



\begin{thebibliography}{BDN15} 

\bibitem[AD17]{AD15}
S.~Alazzawi and W.~Dybalski.
\newblock {Compton scattering in the Buchholz–Roberts framework of relativistic QED}.
\newblock \emph{Lett. Math. Phys.}, 107:81--106, 2017.
\newblock arXiv:\href{https://arxiv.org/abs/1509.03997}{1509.03997}, doi:\href{https://doi.org/10.1007/s11005-016-0889-8}{10.1007/s11005-016-0889-8}.

\bibitem[AG01]{AG01}
S.~Albeverio and H.~Gottschalk.
\newblock Scattering theory for quantum fields with indefinite metric.
\newblock \emph{Commun. Math. Phys.}, 216:491–513, 2001.
\newblock doi:\href{https://doi.org/10.1007/s002200000332}{10.1007/s002200000332}.

\bibitem[A]{ArQFT99}
H.~Araki.
\newblock \emph{Mathematical Theory of Quantum Fields}.
\newblock no.~101, International Series of Monographs on Physics.
\newblock Oxford University Press, 1999.

\bibitem[Arv82]{Arv80}
W.~Arveson.
\newblock {The Harmonic Analysis of Automorphism Groups}.
\newblock In \emph{Operator Algebras and Applications, Part I (Kingston, Ont.,
1980)}, no.~38 in Proc. Sympos. Pure. Math., 199–269. AMS, 1982.
\newblock doi:\href{https://doi.org/10.1090/pspum/038.1}{10.1090/pspum/038.1}.

\bibitem[BDN15]{BDN15}
S.~Bachmann, W.~Dybalski, and P.~Naaijkens.
\newblock {Lieb-Robinson} bounds, {Arveson} spectrum and {Haag–Ruelle} scattering theory for gapped quantum spin systems.
\newblock \emph{Ann. Henri Poincar{é}}, pp.\ 1–55, 2015.
\newblock doi:\href{https://doi.org/10.1007/s00023-015-0440-y}{10.1007/s00023-015-0440-y}.

\bibitem[BW84]{BW84}
H.~Baumgärtel and M.~Wollenberg.
\newblock A class of nontrivial weakly local massive {Wightman} fields with interpolating properties.
\newblock \emph{Commun. Math. Phys.}, 94:331–352, 1984.
\newblock doi:\href{https://doi.org/10.1007/BF01224829}{10.1007/BF01224829}.

\bibitem[BW75]{BW75}
J.~J. Bisognano and E.~H. Wichmann.
\newblock On the duality condition for a hermitian scalar field.
\newblock \emph{J. Math. Phys.}, 16:985--1007, 1975.
\newblock doi:\href{https://doi.org/10.1063/1.522605}{10.1063/1.522605}.

\bibitem[Bor95]{Bor95}
H.-J. Borchers.
\newblock When does {Lorentz} invariance imply wedge duality?
\newblock \emph{Lett. Math. Phys.}, 35:39–60, 1995.
\newblock doi:\href{https://doi.org/10.1007/BF00739154}{10.1007/BF00739154}.

\bibitem[BBS01]{BBS01}
H.-J. Borchers, D.~Buchholz, and B.~Schroer.
\newblock Polarization-free generators and the {S}-matrix.
\newblock \emph{Commun. Math. Phys.}, 219:125–140, 2001.
\newblock doi:\href{https://doi.org/10.1007/s002200100411}{10.1007/s002200100411}.

\bibitem[BR87]{BR87}
O.~Bratteli and D.W. Robinson.
\newblock \emph{Operator Algebras and Quantum Statistical Mechanics 1}.
\newblock Springer, 1987.
\newblock doi:\href{https://doi.org/10.1007/978-3-662-02520-8}{10.1007/978-3-662-02520-8}.

\bibitem[Bu77]{Bu77}
D.~Buchholz.
\newblock {Collision Theory for Massless Bosons}.
\newblock \emph{Commun. Math. Phys.}, 52:147–173, 1977.
\newblock doi:\href{https://doi.org/10.1007/BF01625781}{10.1007/BF01625781}.

\bibitem[Bu17]{Bu17pc}
D.~Buchholz.
\newblock private communications, 2017.

\bibitem[BLS11]{BLS10}
D.~Buchholz, G.~Lechner, and S.J. Summers.
\newblock {Warped Convolutions, Rieffel Deformations and the Construction of Quantum Field Theories}.
\newblock \emph{Commun. Math. Phys.},\,304:95--123,\,2011.\,%
\newblock arXiv:\href{https://arxiv.org/abs/1005.2656}{1005.2656},\,%
doi:\href{https://doi.org/10.1007/s00220-010-1137-1}{10.1007/s00220-010-1137-1}.

\bibitem[BS08]{BS08}
D.~Buchholz and S.~J. Summers.
\newblock {Warped Convolutions:. a Novel Tool in the Construction of Quantum Field Theories}.
\newblock In E.~{Seiler} et al. (ed.), \emph{Quantum Field Theory and Beyond. Essays in Honor of Wolfhart Zimmermann}, 107--121. 2008.
\newblock arXiv:\href{https://arxiv.org/abs/0806.0349}{0806.0349}, doi:\href{https://doi.org/10.1142/9789812833556_0007}{10.1142/9789812833556\_0007}.

\bibitem[DFR95]{DFR95}
S.~Doplicher, K.~Fredenhagen, and J.~E. Roberts.
\newblock {The Quantum structure of space-time at the Planck scale and quantum fields}.
\newblock \emph{Commun. Math. Phys.}, 172:187--220, 1995.
\newblock doi:\href{https://doi.org/10.1007/BF02104515}{10.1007/BF02104515}.

\bibitem[Du17]{MD16}
M.~Duell.
\newblock Strengthened {Reeh-Schlieder} property and scattering in quantum field theories without mass gaps.
\newblock \emph{Commun. Math. Phys.}, 352:935–966, 2017.
\newblock arXiv:\href{https://arxiv.org/abs/1603.07512}{1603.07512}, doi:\href{https://doi.org/10.1007/s00220-017-2841-x}{10.1007/s00220-017-2841-x}.

\bibitem[Dy05]{Dy05}
W.~Dybalski.
\newblock {Haag-Ruelle scattering theory in presence of massless particles}.
\newblock \emph{Lett. Math. Phys.}, 72:27–38, 2005.
\newblock arXiv:\href{https://arxiv.org/abs/hep-th/0412226}{hep-th/0412226}, doi:\href{https://doi.org/10.1007/s11005-005-2294-6}{10.1007/s11005-005-2294-6}.

\bibitem[DG14]{DG13}
W.~Dybalski and C.~Gérard.
\newblock {A criterion for asymptotic completeness in local relativistic QFT}.
\newblock \emph{Commun. Math. Phys.}, 332:1167--1202, 2014.
\newblock arXiv:\href{https://arxiv.org/abs/1308.5187}{1308.5187}, doi:\href{https://doi.org/10.1007/s00220-014-2069-y}{10.1007/s00220-014-2069-y}.

\bibitem[DT11]{DT10}
W.~Dybalski and Y.~Tanimoto.
\newblock {Asymptotic completeness in a class of massless relativistic quantum field theories}.
\newblock \emph{Commun. Math. Phys.}, 305:427--440, 2011.
\newblock arXiv:\href{https://arxiv.org/abs/1006.5430}{1006.5430}, doi:\href{https://doi.org/10.1007/s00220-010-1173-x}{10.1007/s00220-010-1173-x}.

\bibitem[FGR96]{FGR96}
K.~Fredenhagen, M.~R. Gaberdiel, and S.~M. Rüger.
\newblock Scattering states of plektons (particles with braid group statistics) in 2+1 dimensional quantum field theory.
\newblock \emph{Commun. Math. Phys.}, 175:319–335, 1996.
\newblock doi:\href{https://doi.org/10.1007/BF02102411}{10.1007/BF02102411}.

\bibitem[GL07]{GL07}
H.~Grosse and G.~Lechner.
\newblock Wedge-local quantum fields and noncommutative {Minkowski} space.
\newblock \emph{JHEP}, 2007:12, 2007.
\newblock arXiv:\href{https://arxiv.org/abs/0706.3992}{0706.3992}, doi:\href{https://doi.org/10.1088/1126-6708/2007/11/012}{10.1088/1126-6708/2007/11/012}.

\bibitem[Ha58]{Ha58}
R.~Haag.
\newblock Quantum field theories with composite particles and asymptotic conditions.
\newblock \emph{Phys. Rev.}, 112:669–673, 1958.
\newblock doi:\href{https://doi.org/10.1103/PhysRev.112.669}{10.1103/PhysRev.112.669}.

\bibitem[Hep65]{Hep65}
K.~Hepp.
\newblock On the connection between the {LSZ} and {Wightman} quantum field theory.
\newblock \emph{Commun. Math. Phys.}, 1:95--111, 1965.
\newblock doi:\href{https://doi.org/10.1007/BF01646494}{10.1007/BF01646494}.

\bibitem[Her13]{Hdg13}
A.~Herdegen.
\newblock {Infraparticle Problem, Asymptotic Fields and Haag–Ruelle Theory}.
\newblock \emph{Ann. H. Poincar{é}}, 15:345–367, 2013.
\newblock doi:\href{https://doi.org/10.1007/s00023-013-0242-z}{10.1007/s00023-013-0242-z}.

\bibitem[Le03]{Le03}
G.~Lechner.
\newblock {Polarization free quantum fields and interaction}.
\newblock \emph{Lett. Math. Phys.}, 64:137--154, 2003.
\newblock doi:\href{https://doi.org/10.1023/A:1025772304804}{10.1023/A:1025772304804}.

\bibitem[Le06]{Le06t}
G.~Lechner.
\newblock \emph{{On the construction of quantum field theories with factorizing S-matrices}}.
\newblock PhD thesis, Univ.\ Göttingen, 2006.
\newblock arXiv:\href{https://arxiv.org/abs/math-ph/0611050}{math-ph/0611050}.

\bibitem[Le15]{Le15}
G.~Lechner.
\newblock Algebraic constructive quantum field theory: integrable models and deformation techniques.
\newblock In \emph{Advances in Algebraic Quantum Field Theory}, 397–448. Springer, 2015.
\newblock arXiv:\href{https://arxiv.org/abs/1503.03822}{1503.03822}, doi:\href{https://doi.org/10.1007/978-3-319-21353-8_10}{10.1007/978-3-319-21353-8\_10}.

\bibitem[LTU17]{LTU17}
R.~Longo, Y.~Tanimoto, and Y.~Ueda.
\newblock {Free products in AQFT}.
\newblock 2017.
\newblock arXiv:\href{https://arxiv.org/abs/1706.06070}{1706.06070}.

\bibitem[Mo17]{Mo17}
V.~Morinelli.
\newblock {The Bisognano-Wichmann property on nets of standard subspaces, some sufficient conditions}.
\newblock 2017.
\newblock arXiv:\href{https://arxiv.org/abs/1703.06831}{1703.06831}.

\bibitem[RS2]{RS2}
M.~Reed and B.~Simon.
\newblock \emph{{Methods of Modern Mathematical Physics, Vol. 2, Fourier Analysis and Self-Adjointness}}.
\newblock Academic Press, 1975.

\bibitem[Ru62]{Ru62}
D.~Ruelle.
\newblock On the asymptotic condition in quantum field theory.
\newblock \emph{Helv. Phys. Acta}, 35:147–163, 1962.
\newblock doi:\href{https://doi.org/10.5169/seals-113272}{10.5169/seals-113272}.

\bibitem[Smi92]{Smi92}
F.A. Smirnov.
\newblock \emph{Form Factors in Completely Integrable Models of Quantum Field Theory}.
\newblock Adv. Series in Mathematical Physics.
\newblock World Scientific, 1992.

\end{thebibliography}
\end{document}